\documentclass[lettersize,journal,10pt]{IEEEtran}
\IEEEoverridecommandlockouts
\usepackage{cite}
\usepackage{amsmath,amssymb,amsfonts}

\usepackage{amsthm}
\usepackage{mathrsfs}
\usepackage{bm}
\usepackage[ruled,vlined]{algorithm2e}
\usepackage{subfigure}

\usepackage{graphicx}
\usepackage{textcomp}
\usepackage{multirow}
\usepackage{threeparttable}

\usepackage{stfloats}
\graphicspath{{figures/}}
\newcommand{\tabincell}[2]{\begin{tabular}{@{}#1@{}}#2\end{tabular}}

\usepackage{hyperref}

\newtheorem{theorem}{Theorem}

\newtheorem{remark}{Remark}
\usepackage{url}

\makeatletter
\def\url@leostyle{%
	\@ifundefined{selectfont}{\def\UrlFont{\sf}}{\def\UrlFont{\small\ttfamily}}}
\makeatother
\begin{document}

\title{Fast List Decoding of High-Rate Polar Codes}

\author{\IEEEauthorblockN{Yang Lu, Ming-Min Zhao,~\IEEEmembership{Member,~IEEE}, Ming Lei,~\IEEEmembership{Member,~IEEE}, \\ and Min-Jian Zhao,~\IEEEmembership{Senior Member,~IEEE}}
\thanks{The authors are with the College of Information Science and Electronic Engineering, Zhejiang University, Hangzhou 310027, China (email: \{22031097, zmmblack, lm1029, mjzhao\}@zju.edu.cn).}}

\maketitle

\begin{abstract} Due to the ability to provide superior error-correction performance, the successive cancellation list (SCL) algorithm is widely regarded as one of the most promising decoding algorithms for polar codes with short-to-moderate code lengths. However, the application of SCL decoding in low-latency communication scenarios is limited due to its sequential nature. To reduce the decoding latency, developing tailored fast and efficient list decoding algorithms of specific polar substituent codes (special nodes) is a promising solution. Recently, fast list decoding algorithms are proposed by considering special nodes with low code rates. Aiming to further speedup the SCL decoding, this paper presents fast list decoding algorithms for two types of high-rate special nodes, namely single-parity-check (SPC) nodes and sequence rate one or single-parity-check (SR1/SPC) nodes. In particular, we develop two classes of fast list decoding algorithms for these nodes, where the first class uses a sequential decoding procedure to yield decoding latency that is linear with the list size, and the second further parallelizes the decoding process by pre-determining the redundant candidate paths offline. Simulation results show that the proposed list decoding algorithms are able to achieve up to 70.7\% lower decoding latency than state-of-the-art fast SCL decoders, while exhibiting the same error-correction performance.
\end{abstract}

\section{Introduction}
\label{section:1}
\IEEEPARstart{P}{roved} to achieve the symmetric capability of memoryless channels \cite{Arikan2009Channel}, polar codes have been adopted as the control channel coding scheme for the enhanced mobile broadband (eMBB) scenario in the latest 5G cellular communication standard \cite{5Gstandard}. The very first decoding algorithm for polar codes is successive cancellation (SC) \cite{Arikan2009Channel}, which is able to approach the maximum likelihood (ML) decoding performance as the code length tends towards infinity. However, the SC decoding performance degrades significantly for practical short-to-moderate code lengths. To narrow the performance gap between SC and ML, alternative decoding algorithms have been proposed, among which SC list (SCL) is one of the most promising \cite{Tal2015List}. By maintaining a list of more reliable candidate codewords during the decoding process, SCL decoding provides much better error-correction performance than SC, at the cost of increased implementation complexity. When employing cyclic redundancy check (CRC) codes as a genie to select the most probable codeword from the list, the error-correction performance of the SCL decoder can be further improved \cite{Niu2012CRC}. With a large list size, it was shown that the CRC-aided SCL (CA-SCL) decoder can achieve near-ML performance, making polar codes competitive with the state-of-the-art low-density parity-check (LDPC) and turbo codes \cite{Balatsoukas2017Comparison}.

In spite of the promising error-correction performance, both SC and SCL suffer from high decoding latency due to the inherent sequential bit-by-bit decision process. To facilitate their applications in low-latency use cases, various works have been proposed to develop parallel decoding techniques for polar codes \cite{Alamdar2011Simplified,Sarkis2013Increasing,Sarkis2014Fast,Hanif2017Fast,Condo2018Generalized,Sarkis2016Fast,Hashemi2016Fast,Hashemi2017Fast,Ardakani2019Fast,Zhao2021Minimum,Zheng2021Threshold,Ren2022Sequence,Lu2023Fast_Conf,Lu2023Fast}. The main advantage of these schemes is that they significantly reduce the time complexity of the conventional SC and SCL decoders with reduced or slightly increased space complexity. In particular, these schemes identify various kinds of special polar constituent codes (special nodes) which have specific frozen and information bit patterns, and implement efficient node-based decoders to decode these bits in parallel. For instance, it was discovered in \cite{Alamdar2011Simplified,Sarkis2013Increasing,Sarkis2014Fast} that rate zero (R0) nodes without information bits, rate one (R1) nodes without frozen bits, repetition (REP) nodes with a single information bit in the rightmost position, and single parity-check (SPC) nodes with a single frozen bit in the leftmost position, can be decoded in parallel using low-complexity decoding algorithms. As compared to the conventional SC decoding algorithm, these highly-parallel decoding algorithms are able to simplify the decoding process and reduce the latency significantly, without incurring any decoding performance degradation. Likewise, fast SC decoders were later presented for five new types of special nodes, namely the Type I-V nodes, thus further decreasing the SC decoding latency \cite{Hanif2017Fast}. 

However, all these works require separate decoders for each node type, which inevitably complicates the hardware implementation. In \cite{Condo2018Generalized}, two types of more general nodes, i.e., the generalized REP (G-REP) node and the generalized parity-check (G-PC) node, were proposed to provide unified descriptions of existing special nodes, such as Type I-V nodes. Moreover, with tailored fast decoding algorithms for G-REP and G-PC nodes, the SC decoding latency can be further reduced. Recently, the authors in \cite{Zheng2021Threshold} identified a class of sequence R0 or REP (SR0/REP) nodes which envelops most of the aforementioned low-rate special nodes as special cases. Based on exhaustive search of the information bits in REP nodes, a fast decoding algorithm for the SR0/REP node was proposed to achieve a higher degree of parallelism without degrading the error-correction performance. Although the decoding of SR0/REP nodes guarantees further latency reduction with respect to the previous works, it comes at the cost of increased hardware resource consumption \cite{Zheng2020Implementation}. 

The identification and utilization of the aforementioned special nodes were also extended to SCL decoding \cite{Hashemi2016Fast,Hashemi2017Fast,Condo2018Generalized,Ardakani2019Fast,Zhao2021Minimum,Ren2022Sequence}. Specifically, a simplified SCL (SSCL) decoder was first proposed in \cite{Hashemi2016Fast}, which is able to simplify the decoding of R0, R1, REP and SPC nodes while preserving the error-correction performance of the conventional SCL decoder. This work was later advanced in \cite{Hashemi2017Fast}, where it was shown that redundant path splitting processes associated with a specific list size can be avoided for R1 and SPC nodes. By applying this optimized path splitting strategy to the SSCL decoder, the resultant Fast-SSCL decoder yields exactly the same error-correction performance yet with reduced decoding latency and computational complexity. Moreover, the works \cite{Condo2018Generalized} and \cite{Ardakani2019Fast} extended the above fast SCL decoding techniques by taking the G-PC/G-REP and Type I-V nodes into consideration, respectively. A list decoding algorithm for SR0/REP nodes was presented in \cite{Ren2022Sequence}, which was able to increase the throughput and reduce the latency significantly.

For high-rate special nodes, including R1, SPC and Type III-IV, etc., the mainstream fast list decoders share a similar sequential procedure based on sphere decoding \cite{Hashemi2016Fast}, where path splitting and pruning are performed step by step at the node level. Although the traversal of the whole decoding tree is avoided, these fast sequential list (FSL) decoders fall short in providing higher degrees of parallelism due to their sequential nature. In \cite{Zhao2021Minimum}, the authors showed that the sequential path splitting process for R1 nodes \cite{Hashemi2017Fast} can be further simplified and parallelized to achieve considerable latency reduction. Specifically, by pre-collecting all flipping bit index combinations that will certainly lead to redundant candidate paths, a minimum-combinations set (MCS) can be determined offline. Based on the MCS, a fast parallelized list (FPL) decoder was presented for R1 nodes, which exhibits considerable speedup when compared to the Fast-SSCL decoder, without any performance degradation.

To achieve further latency reduction, it is critical to parallelize the fast list decoders for other high-rate special nodes with more complicated parity constraints. Recently, another generalized sequence node composed of a sequence of R1 or SPC nodes (SR1/SPC) was discovered to envelop most of the existing high-rate special nodes \cite{Lu2023Fast_Conf,Lu2023Fast}. To decode SR1/SPC nodes efficiently, the authors analysed parity constraints caused by the frozen bits, leading to a validity rule that should be satisfied during the decoding process. Based on the validity rule and ML rule, a highly parallelized decoding algorithm was finally proposed for SR1/SPC nodes, which further reduces the SC decoding latency significantly. In this paper, we take one step further by presenting fast list decoding algorithms for SPC and SR1/SPC nodes, such that a large range of high-rate polar constitution codes can be decoded more efficiently. Our contributions are summarized as follows:
\begin{itemize}
	\item For SPC nodes, we show how to generate the MCS offline when different parity checks are taken into consideration, such that a large number of redundant candidate paths can be pre-determined and ignored during the decoding process. Based on the MCS, the conventional FSL decoder can be highly parallelized, resulting in a FPL decoder without any degradation in error-correction performance.
	\item For SR1/SPC nodes, two classes of fast list decoding algorithms, namely FSL and FPL, are presented to decode the descendant nodes. Specifically, the proposed algorithms mainly consist of two stages, where stage I is devoted to decoding the source node and stage II decodes the whole special node based on the source node decoding results. These two decoding algorithms can reduce the decoding latency to varying degrees, providing a flexible tradeoff between decoding latency and complexity. 
	\item Detailed decoding latency analysis is presented to show the superiority of the proposed decoders. Besides, simulation results show that implementing the proposed decoders can improve the parallelism of the conventional fast SCL decoders significantly, up to 70.7\% decoding latency reduction can be achieved as compared to the state-of-the-art (SOTA) fast SCL decoders, without any performance degradation.
\end{itemize}

The remainder of this paper is organized as follows. Section~\ref{sec2} reviews the backgrounds on polar codes, SC/SCL decoding, and fast decoding techniques. In Section~\ref{sec3}, we present a fast list decoding algorithm for SPC nodes. Section~\ref{sec4} provides two list decoding algorithms for SR1/SPC nodes. Section~\ref{sec5} presents several a detailed decoding latency analysis, empirical optimization methods for the proposed decoders, and provides simulation results to evaluate the decoding performance and latency. Finally, conclusions are drawn in Section~\ref{sec6}.

\section{Preliminaries}
\label{sec2}
\emph{Notations}: Scalars, vectors, and matrices are respectively denoted by lower case, boldface lower case, and boldface upper case letters. $\mathcal{N}$ represents the natural number set. $\%$ denotes the mod operation, $\operatorname{sgn}(a)$ denotes the sign of a scalar $a$ and $\min(\bm{a})$ returns the minimum element in $\bm{a}$. For two arbitrary sets $\mathcal{A}$ and $\mathcal{B}$, $\triangle$ denotes the XOR operation such that $\mathcal{A} \triangle \mathcal{B} = (\mathcal{A}-\mathcal{B}) \cup (\mathcal{B}-\mathcal{A})$. Besides, $\oplus$ denotes the bitwise XOR operation and $\otimes$ denotes the Kronecker product. $\bm{A}^{\otimes n}$ denotes the $n$-th Kronecker power of $\bm{A}$.

\vspace{-1em}
\subsection{Polar Codes}
A polar code with code length $N=2^{n}$ and information length $K$, denoted by $\mathcal{P}(N, K)$, maps a message vector $\bm{u} = (u[1], u[2], \ldots, u[N])$ into a polar codeword $\bm{x} = (x[1], x[2], \ldots, x[N])$ by using a linear transformation $\bm{x} = \bm{u} \bm{G}_N$, where $\bm{G}_N = \bm{F}^{\otimes n}$ is the generator matrix with $\bm{F} = \begin{bmatrix}\begin{smallmatrix}	1 & 0 \\ 1 & 1 \end{smallmatrix}\end{bmatrix}$ being the base kernel. The principle of channel polarization reveals that the bits in $\bm{u}$ corresponds to bit-channels with different reliabilities \cite{Arikan2009Channel}. Amongst the $N$ bit-channels, the $K$ most reliable ones are chosen to transmit information bits, while the remaining $N-K$ ones are assigned with frozen bits (usually set to 0). The information and frozen bit-channels can be distinguished with an indicator vector $\bm{c} = (c[1], c[2], \ldots, c[N])$, i.e.,
\begin{equation}
	\begin{aligned}
		c[i] = \begin{cases} 1, & \textrm{if}\ i \in \mathcal{I} \\ 0, & \textrm{if}\ i \in \mathcal{I}^c\end{cases},
	\end{aligned}
	\label{eqn:1}
\end{equation}
where $\mathcal{I}$ and $\mathcal{I}^c$ denote the sets of information and frozen bit indices, respectively, which are both known to the encoder and decoder.

After encoding, the codeword vector $\bm{x}$ is then modulated and transmitted over the channel. Throughout this paper, we consider binary phase shift keying (BPSK) modulation and additive white Gaussian noise (AWGN) channel. After receiving $\bm{x}$ from the channel, the receiver provides log-likelihood ratio (LLR) of the received bits to the polar decoder, and an estimate of the original message $\hat{\bm{u}}$ is obtained.
\vspace{-1em}
\subsection{SC and SCL Decoding}
SC and SCL decoding can be interpreted as a binary tree traversal which starts from the root node to the leaf node and from the left branch to the right. At the $s$-th level of the decoding tree, the $i$-th node passes soft information, i.e., the LLR vector $\alpha_{s,i}[1:2^s]$, to its child nodes, whereas it also receives hard information, i.e., the estimated codeword $\beta_{s,i}[1:2^s]$, from its child nodes in return, where $1 \leq i \leq 2^{n-s}$ and $0 \leq s \leq n$. In particular, the LLRs are updated by
\begin{subequations}
\begin{gather}
	\begin{split}
		\alpha_{s-1,2i-1}[k] =& \operatorname{sgn}(\alpha_{s,i}[k]) \operatorname{sgn}(\alpha_{s,i}[k+2^{s-1}]) \\
		& \min(|\alpha_{s,i}[k]|, |\alpha_{s,i}[k+2^{s-1}]|), 
	\end{split} \\
	\begin{split}
		\alpha_{s-1,2i}[k] = & (1-2\beta_{s-1,2i-1}[k])\alpha_{s,i}[k] \\
		& +\alpha_{s,i}[k+2^{s-1}], 
	\end{split}
\end{gather}
\label{eqn:2}
\end{subequations}
whereas the codeword is updated according to
\begin{equation}
\begin{aligned}
\beta_{s,i}[k] = \begin{cases}\beta_{s-1,2i-1}[k] \oplus \beta_{s-1,2i}[k], & \textrm{if}\ 1 \leq k \leq 2^{s-1} \\
\beta_{s-1,2i}[k], & \textrm{otherwise}\end{cases}.
\end{aligned}
\label{eqn:3}
\end{equation}

At the leaf level $s=0$, SC decoding chooses a locally optimal estimate, i.e., the hard decision output, of each information bit, which is shown as follows:
\begin{equation}
\begin{aligned}
\hat{u}[i] = \operatorname{HD}(\alpha[i]) =  \begin{cases}\frac{1-\textrm{sgn}(\alpha[i])}{2}, & \textnormal{if $i \in \mathcal{I}$} \\ 0, & \textnormal{if $i \in \mathcal{I}^c$}\end{cases},
\end{aligned}
\label{eqn:4}
\end{equation}
where $\hat{u}[i]$ and $\alpha[i]$ are the estimate and LLR of $u[i]$, respectively, and $\operatorname{HD}(\cdot)$ is the hard decision function.

Instead of keeping only one estimated codeword (path), SCL decoding allows for maintaining a list of up to $L$ candidate paths by considering both hypotheses, i.e., $u[i]=0$ and $u[i]=1$ for each information bit (also known as path splitting). As such, the number of candidate paths will double after decoding each information bit, until it exceeds the list size $L$, then only the $L$ most reliable candidate paths are reserved for the subsequent decoding (i.e., path pruning). In order to evaluate the reliability of each path, a path metric (PM), denoted by $\operatorname{PM}^l_{i}$, was introduced in \cite{Balatsoukas2015LLR}, which is obtained by
\begin{equation}
	\begin{aligned}
		\operatorname{PM}^l_{i} = & \sum_{k=1}^{i}\ln(1+e^{-(1-2\hat{u}^l[k])\alpha^l[k]}) \\
		=& \operatorname{PM}^l_{i-1} + \ln(1+e^{-(1-2\hat{u}^l[i])\alpha^l[i]}),
	\end{aligned}
	\label{eq5}
\end{equation}
where the superscript $l$ indicates the $l$-th path, and $\ln(1+e^{-(1-2\hat{u}^l[i])\alpha^l[i]})$ can be viewed as the penalty caused by the mismatch between the estimation and hard-decision result of $u[i]$, given by
\begin{equation}
	\begin{aligned}
		&\ln(1+e^{-(1-2\hat{u}^l[i])\alpha^l[i]})\\
		&=\begin{cases} \ln(1+e^{-|\alpha^l[i]|}), & \textnormal{if $\hat{u}^l[i] = \operatorname{HD}(\alpha^{l}[i])$} \\ \ln(1+e^{|\alpha^l[i]|}), & \textnormal{otherwise}\end{cases}.
	\end{aligned}
	\label{eq6}
\end{equation}
Using the following hardware-friendly (HWF) approximation:
\begin{equation*}
	\begin{aligned}
		 \ln (1+e^a) = \begin{cases} a, & \textnormal{if $a > 0$} \\ 0, & \textnormal{otherwise}\end{cases},
	\end{aligned}
\end{equation*}
the calculation of PM can be simplified as \cite{Balatsoukas2015LLR}
\begin{equation}
	\begin{aligned}
		\operatorname{PM}^l_{i} 
		\approx& \begin{cases} \operatorname{PM}^l_{i-1}, & \textnormal{if $\hat{u}^l[i] = \operatorname{HD}(\alpha^{l}[i])$} \\ \operatorname{PM}^l_{i-1} + |\alpha^l[i]|, & \textnormal{otherwise}\end{cases}.
	\end{aligned}
	\label{eq7}
\end{equation}
Without loss of generalization, we assume that the candidate paths are sorted in ascending order in terms of PM values. Besides, for clarity, we refer to the $L$ paths before decoding a bit or a node as \emph{parent paths}, the paths generated by splitting the parent paths as \emph{candidate paths}, the $L$ paths with the smallest PMs amongst all candidate paths as \emph{reserved paths}, and all the remaining candidate paths that are eliminated as \emph{redundant paths}. Therefore, after decoding a bit or a node, the previously reserved paths are actually the parent paths for the subsequent decoding. In the following, we omit the superscript $l$ for brevity.
\vspace{-0.5em}
\subsection{Special Nodes}
\begin{table*}[ht]
	\caption{Structures of Different Special Nodes}
	\centering
	\begin{threeparttable}
		\begin{tabular}{c|c|c|c}\hline
			R0 & $\bm{c} = (0, 0, \ldots, 0)$ & R1 & $\bm{c} = (1, 1, \ldots, 1)$ \\ \hline
			REP & $\bm{c} = (0, \ldots, 0, 1)$ & SPC & $\bm{c} = (0, 1, \ldots, 1)$ \\ \hline
			Type-I & $\bm{c} = (0, \ldots, 0, 1, 1)$ & Type-II & $\bm{c} = (0, \ldots, 0, 1, 1, 1)$ \\ \hline
			Type-III & $\bm{c} = (0, 0, 1, \ldots, 1)$ & Type-IV & $\bm{c} = (0, 0, 0, 1, \ldots, 1)$ \\ \hline
			Type-V & $\bm{c} = (0, \ldots, 0, 1, 0, 1, 1, 1)$ & G-REP & $\bm{c} = (\overbrace{0,\ldots,0}^{N_{p-1}},\overbrace{0,\ldots,0}^{N_{p-2}},\ldots,\overbrace{0,\ldots,0}^{N_q} \llap{$\underbrace{\phantom{0,\ldots,0,0,\ldots,0,\ldots,0,\ldots,0}}_{\textnormal{R0}}$},\overbrace{\textrm{X},\ldots,\textrm{X}}^{N_q} \llap{$\underbrace{\phantom{\textrm{X},\ldots,\textrm{X}}}_{\textnormal{source node}}$})$ \\ \hline
			G-PC & $\bm{c} = (\overbrace{0,\ldots,0}^{N_q} \llap{$\underbrace{\phantom{0,\ldots,0}}_{\textnormal{R0}}$},\overbrace{1,\ldots,1}^{N_{q+1}},\overbrace{1,\ldots,1}^{N_{q+2}},\ldots,\overbrace{1,\ldots,1}^{N_{p-1}} \llap{$\underbrace{\phantom{1,\ldots,1,1,\ldots,1,\ldots,1,\ldots,1}}_{\textnormal{R1}}$})$ & EG-PC & $\bm{c} = (\overbrace{0,\ldots,0,\textrm{X}}^{N_q} \llap{$\underbrace{\phantom{0,\ldots,0,\textrm{X}}}_{\textnormal{R0 or REP}}$},\overbrace{1,\ldots,1}^{N_{q+1}},\overbrace{1,\ldots,1}^{N_{q+2}},\ldots,\overbrace{1,\ldots,1}^{N_{p-1}} \llap{$\underbrace{\phantom{1,\ldots,1,1,\ldots,1,\ldots,1,\ldots,1}}_{\textnormal{R1}}$})$ \\ \hline
			SR0/REP & $\bm{c} = (\overbrace{0,\ldots,0,\textrm{X}}^{N_{p-1}},\overbrace{0,\ldots,0,\textrm{X}}^{N_{p-2}},\ldots,\overbrace{0,\ldots,0,\textrm{X}}^{N_q} \llap{$\underbrace{\phantom{0,\ldots,0,\textrm{X},0,\ldots,0,\textrm{X},\ldots,0,\ldots,0,\textrm{X}}}_{\textnormal{R0 or REP}}$},\overbrace{\textrm{X},\ldots,\textrm{X}}^{N_q} \llap{$\underbrace{\phantom{\textrm{X},\ldots,\textrm{X}}}_{\textnormal{source node}}$})$ & SR1/SPC & $\bm{c} = (\overbrace{\textrm{X},\ldots,\textrm{X}}^{N_q} \llap{$\underbrace{\phantom{\textrm{X},\ldots,\textrm{X}}}_{\textnormal{source node}}$},\overbrace{\textrm{X},1,\ldots,1}^{N_{q+1}},\overbrace{\textrm{X},1,\ldots,1}^{N_{q+2}},\ldots,\overbrace{\textrm{X},1,\ldots,1}^{N_{p-1}} \llap{$\underbrace{\phantom{\textrm{X},1,\ldots,1,\textrm{X},1,\ldots,1,\ldots,\textrm{X},1,\ldots,1}}_{\textnormal{R1 or SPC}}$})$ \\ \hline
		\end{tabular}
		\begin{tablenotes}  
			\footnotesize
			\item[1] $\textrm{X}$ indicates either an information or a frozen bit. 
			\item[2] $N_s=2^s$, where $q \leq s < p$.
		\end{tablenotes} 
	\end{threeparttable}
	\label{tab1}
	\vspace{-1em}
\end{table*}

The sequential nature of SC-based decoding, i.e., each bit estimate depends on all previous ones, results in high decoding latency. However, some special nodes in the decoding tree have specific frozen and information bit patterns, which allows for directly obtaining the estimated codewords using tailored decoders. For clarity, Table~\ref{tab1} lists most of the existing special nodes \cite{Alamdar2011Simplified,Sarkis2014Fast,Hanif2017Fast,Condo2018Generalized,Zheng2021Threshold,Lu2023Fast} along with their structure descriptions.

\begin{figure*}[ht]
	\centering
	\includegraphics[width=0.75\textwidth]{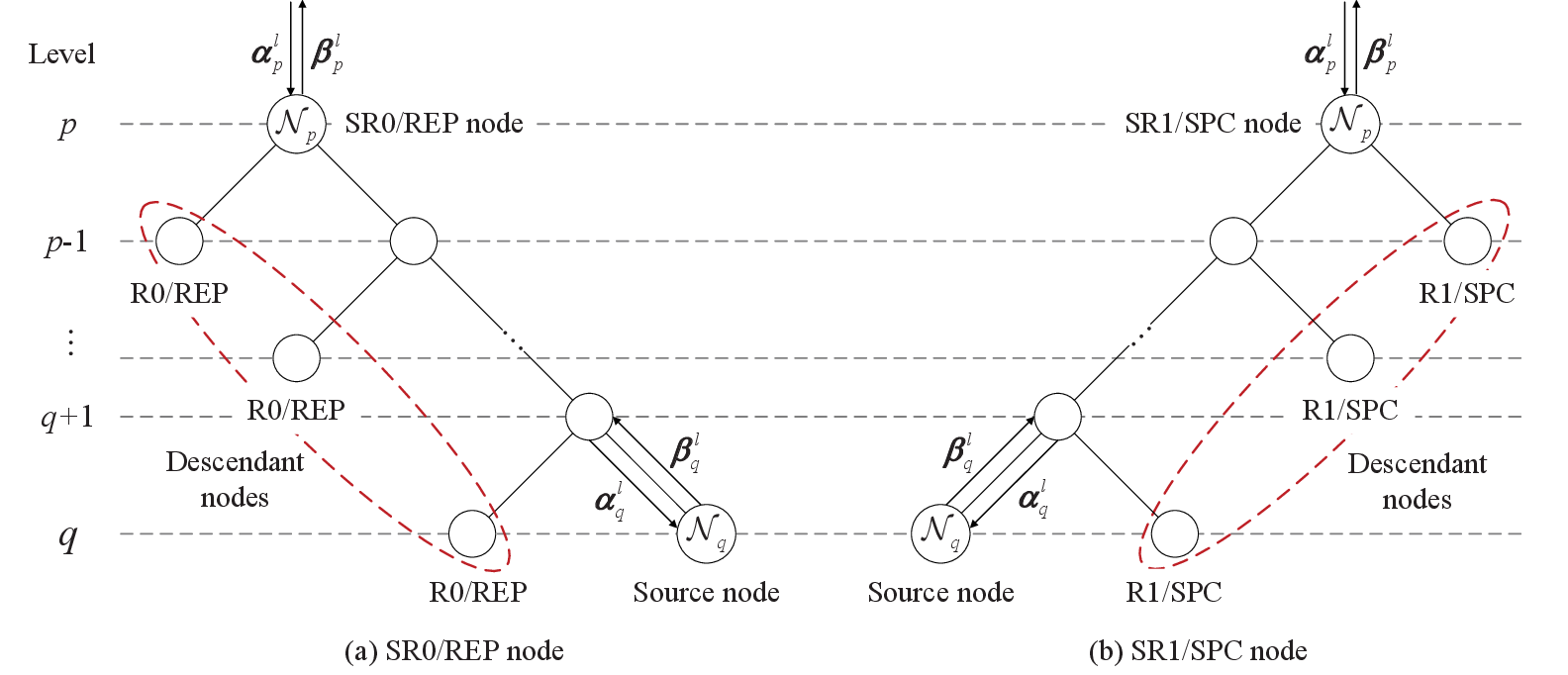}
	\caption{General tree structure of (a) SR0/REP node and (b) SR1/SPC node.}
	\label{fig:1}
	\vspace{-1em}
\end{figure*}

Recently, the sequence nodes, i.e., SR0/REP \cite{Zheng2021Threshold} and SR1/SPC \cite{Lu2023Fast}, whose tree structures are depicted in Fig.~\ref{fig:1}, are newly proposed as the most generalized special nodes so far. Given the root node level $p$, the SR0/REP node consists of a sequence of R0 or REP nodes (descendant nodes) with each located at level $q\leq s< p$, followed by a source node located at level $q$. Likewise, the SR1/SPC node can be interpreted as a node composed of a source node followed by a sequence of R1 or SPC nodes. In particular, the source node of the sequence nodes can be any generic nodes. As compared to the G-REP and G-PC nodes, the sequence nodes envelop more descendant node types and therefore provide more unified descriptions of polar constituent codes. It was shown in \cite{Zheng2021Threshold} and \cite{Lu2023Fast} that the SR0/REP and SR1/SPC nodes are more frequently distributed in low- and high-rate polar codes, respectively, which shows great potential in achieving a higher degree of decoding parallelism and simplifying the practical implementation of polar codes.

For clarity, we differentiate some special cases of the SR1/SPC node by using different notations. Let $\mathcal{L}$ and $\mathcal{L}^c=\{q,q+1,\ldots,p-1\}-\mathcal{L}$ denote two sets that respectively record the level indices of the descendant SPC and R1 nodes. If the descendant nodes are all R1 nodes, i.e., $\mathcal{L}=\emptyset$, then an SR1/SPC node reduces to a \emph{sequence R1 (SR1) node}. In other cases, an SR1/SPC node is referred to as a \emph{relaxed SR1 (RSR1) node}.

\subsection{Fast Decoding of Sequence Nodes}
\label{sec3.4}
Since the sequence nodes are highly-general polar constituent codes, investigating their efficient decoding algorithms is able to alleviate the high-latency problem, as well as simplify the practical implementation. In \cite{Zheng2021Threshold}, the SR0/REP node is decoded using a hybrid ML-SC decoding algorithm, which first exhaustively estimate the information bits in the descendant REP nodes, and then decode the source node using the conventional SC decoding. Therefore, the resulting latency reduction mainly stems from the saving of sequential LLR update and decoding procedures of the descendant REP nodes.

Although such ML-based methods are highly parallel, they are not applicable to the high-rate SR1/SPC nodes due to the extremely high computational complexity. Alternatively, an efficient fast SC decoding algorithm was presented in \cite{Lu2023Fast}, based on two types of parity constraints, i.e., the parallel parity constraints (P-PC) and the segmental parity constraints (S-PC). For an SR1/SPC node, its source node at level $q$ will impose the following $N_q$ P-PCs at the root node level $p$:
\begin{equation}
	\begin{aligned}
		\bigoplus\limits_{j=1}^{N_p/N_q} \beta_p[(j-1)N_q+k] = \beta_q[k],
	\end{aligned}
	\label{eq9}
\end{equation}
where $1 \leq k \leq N_q$. Specially, the bits $(\beta_p[(j-1)N_q+k])_{0\leq j \leq N_p/N_q}$ involved by the $k$-th P-PC can constitute an \emph{SPC subcode} with a special parity check $\beta_q[k]$. In addition, each descendant SPC node at level $r$ will impose an S-PC at the root node level, which is given by 
\begin{equation}
	\begin{aligned}
		\bigoplus\limits_{j=1}^{N_q/N_r/2} \bigoplus\limits_{k=1}^{N_r} \beta_{p}[(2j-1)N_r+k] = 0. 
	\end{aligned}
	\label{eq10}
\end{equation}


To ensure the validity of the output codeword, these parity constraints should be satisfied simultaneously, otherwise the final SC decoding output must be wrong due to the sequential decoding nature, which will degrade the decoding performance. Besides, the ML decoding rule indicates that the Euclid distance between the decoded codeword and the hard-decision codeword should be as small as possible \cite{Sarkis2013Increasing}. Following the aforementioned decoding rules, the fast decoding of SR1/SPC nodes can be divided into two stages. In the first stage, the P-PCs are corrected by temporarily ignoring the S-PCs, and in the second stage, the S-PCs are corrected without violating the P-PCs, such that all the parity check constraints are satisfied simultaneously. Specifically, Wagner decoding \cite{Silverman1954Coding} is employed to correct the P-PCs and a pre-determined flip coordinate set (FCS) is further presented to correct the S-PCs, such that a list of candidate codewords can be obtained. Furthermore, a penalty metric based on the ML rule is
introduced to measure the reliability of each candidate codeword, with which the least penalised codeword is selected as the decoding output. 

In this paper, we advance the above works by proposing fast list decoding algorithms for high-rate polar constituent codes, i.e., SPC and SR1/SPC nodes, which will be detailed in the following.

\section{Fast List Decoding of SPC Nodes}
\label{sec3}
For high-rate special nodes, it is impractical in terms of computational complexity to exhaust the whole search space of candidate paths, i.e., estimate each information bit and compare all candidate paths to find a list of optimal paths. Alternatively, the mainstream approach is to sequentially perform path splitting and generate candidate paths by flipping bits in the parent paths \cite{Hashemi2016Fast,Hashemi2017Fast,Ardakani2019Fast}. In each time step of such a sequential procedure, the unreliable candidate paths are eliminated from further path splitting, which narrows the search space significantly. In \cite{Zhao2021Minimum}, it is proved that the search space of R1 nodes can be further narrowed down to the range of a MCS. Specifically, by excluding all flipping bit index combinations (flipping combinations) that certainly lead to redundant candidate paths, the remaining combinations (minimum-combinations) in the whole search space constitute the MCS. In essence, the final reserved paths of R1 nodes must be included in the MCS. Therefore, the FSL decoding of R1 nodes can be accelerated significantly by directly selecting the most reliable paths from the candidate paths generated based on the MCS, leading to the FPL decoder.

However, for other special nodes with frozen bits, the structure of the search space is more complicated, since each frozen bit will pose a parity constraint on the codeword. Therefore, the minimum-combinations of R1 nodes are not applicable to other special nodes. For instance, the bits in SPC nodes need to keep an even parity check, while inappropriate flipping combinations may violate the parity constraint. 

In this section, we show how to construct the MCS for SPC nodes, based on which an SPC FPL decoder is presented to achieve lower decoding latency.

\subsection{MCS for SPC Nodes}
First, we introduce some necessary notations for further illustration. Specifically, the bits of an SPC node are sorted according to their reliability, i.e., the absolute values of the LLRs $|\alpha_p[i]|$, with sorted indices $(i)_{|\alpha|}$, such that $|\alpha_p[(1)_{|\alpha|}]|\leq|\alpha_p[(2)_{|\alpha|}]|\leq\ldots\leq|\alpha_p[(N_p)_{|\alpha|}]|$. Given a flipping bit set, the flipping combination $\mathcal{F}$ records the corresponding bit indices. For example, when $\mathcal{F} = \{2,3\}$, the flipping bits are actually the bits $\{\beta_p[(2)_{|\alpha|}],\beta_p[(3)_{|\alpha|}]\}$. Then, the whole search space of R1 nodes, denoted by $\mathcal{F}_{\textnormal{all}}$, is defined as a set composed of all flipping combinations, i.e., $\mathcal{F}_{\textnormal{all}}=\{\mathcal{F}|\mathcal{F} \subseteq \{1,2,\ldots,N_p\}\}$. Moreover, denote $\operatorname{num}(\mathcal{F})$ as a function that counts the number of candidate paths which are more reliable than the candidate path associated with $\mathcal{F}$, then a flipping combination $\mathcal{F}'$ satisfying $\operatorname{num}(\mathcal{F}') \geq L$ indicates that the corresponding candidate path is redundant and can be safely eliminated. Based on the theorems introduced in \cite{Zhao2021Minimum}, the reliability of flipping combinations is compared in groups of two. As such, all reliable minimum-combinations are collected in a set $\mathcal{C}$ which is defined as the MCS, i.e., $\mathcal{C}=\{\mathcal{F}|\operatorname{num}(\mathcal{F}) < L, \mathcal{F}\in\mathcal{F}_{\textnormal{all}}\}$.

For SPC nodes, the XOR results of all bits need to keep an even parity check, i.e., $\bigoplus\nolimits_{i=1}^{N_p}\beta_p[i] = 0$, where $\beta_p[i]=\operatorname{HD}(\alpha_p[i])$. By taking this special parity constraint into consideration, we construct the MCS for SPC nodes as follows. Define the initial parity check, denoted by $\gamma$, as
\begin{equation}
	\begin{aligned}
		\gamma \triangleq \bigoplus\limits_{i=1}^{N_p}\beta_p[i],
	\end{aligned}
	\label{eq11}
\end{equation}
then the flipping rule should depend on the value of $\gamma$. For $\gamma=0$ and $\gamma=1$, the number of flipping bits, i.e., the size of flipping combinations $|\mathcal{F}|$, must be even and odd, respectively. Accordingly, the corresponding search spaces of SPC nodes can be expressed as
\begin{equation}
	\begin{gathered}
		\mathcal{F}_{\textnormal{all}}^{\textnormal{even}}=\{\mathcal{F} | \mathcal{F} \in \mathcal{F}_{\textnormal{all}}, |\mathcal{F}|=2i\}, \\		\mathcal{F}_{\textnormal{all}}^{\textnormal{odd}}=\{\mathcal{F} | \mathcal{F} \in \mathcal{F}_{\textnormal{all}}, |\mathcal{F}|=2i+1\},
	\end{gathered}
	\label{eq12}
\end{equation}
where $i \in \mathcal{N}$. We can further narrow down the search space by restricting the maximum size of flipping combinations, using the following theorem. 

\begin{table*}[ht]
	\caption{MCS $\mathcal{C}$ for Different List Sizes}
	\centering
	\begin{tabular}{|c|c|c|c|c|}\hline
		\multirow{2}{*}{$L$} & \multicolumn{2}{c|}{$\mathcal{C}^{\textnormal{SPC}}$ for SPC} & \multirow{2}{*}{$\mathcal{C}^{\textnormal{R1}}$ for R1 \cite{Zhao2021Minimum}} & \multirow{2}{*}{$|\mathcal{C}|$} \\ \cline{2-3} & $\gamma=0$ & $\gamma=1$ & & \\ \hline

		2 & $\{\emptyset,\{1,2\}\}$ & $\{\{1\},\{2\}\}$ & $\{\emptyset,\{1\}\}$ & 2 \\ \hline
		
		4 & \tabincell{c}{$\{\emptyset,\{1,2\},\{1,3\},\{1,4\},$\\$\{2,3\}\}$} & \tabincell{c}{$\{\{1\},\{2\},\{3\},\{4\},$\\$\{1,2,3\}\}$} & \tabincell{c}{\{$\emptyset,\{1\},\{2\},\{3\},$\\$\{1,2\}\}$} & 5 \\ \hline
		
		8 & \tabincell{c}{$\{\emptyset,\{1,2\},\{1,3\},\{1,4\},$\\$\{1,5\},\{1,6\},\{1,7\},\{1,8\},$\\$\{2,3\},\{2,4\},\{2,5\},$\\$\{3,4\},\{1,2,3,4\}\}$} & \tabincell{c}{$\{\{1\},\{2\},\{3\},\{4\},$\\$\{5\},\{6\},\{7\},\{8\},$\\$\{1,2,3\},\{1,2,4\},\{1,2,5\},$\\$\{1,3,4\},\{2,3,4\}\}$} & \tabincell{c}{$\{\emptyset,\{1\},\{2\},\{3\},$\\$\{4\},\{5\},\{6\},\{7\},$\\$\{1,2\},\{1,3\},\{1,4\},$\\$\{2,3\},\{1,2,3\}\}$} & 13 \\ \hline
		
		16 & \tabincell{c}{$\{\emptyset,\{1,2\},\{1,3\},\{1,4\},$\\$\{1,5\},\{1,6\},\{1,7\},\{1,8\},$\\$\{1,9\},\{1,10\},\{1,11\},\{1,12\},$\\$\{1,13\},\{1,14\},\{1,15\},\{1,16\},$\\$\{2,3\},\{2,4\},\{2,5\},\{2,6\},$\\$\{2,7\},\{2,8\},\{2,9\},$\\$\{3,4\},\{3,5\},\{3,6\},\{3,7\},$\\$\{4,5\},\{4,6\},\{5,6\},$\\$\{1,2,3,4\},\{1,2,3,5\},\{1,2,3,6\},$\\$\{1,2,4,5\},\{1,3,4,5\},\{2,3,4,5\}\}$} & \tabincell{c}{$\{\{1\},\{2\},\{3\},\{4\},$\\$\{5\},\{6\},\{7\},\{8\},$\\$\{9\},\{10\},\{11\},\{12\},$\\$\{13\},\{14\},\{15\},\{16\},$\\$\{1,2,3\},\{1,2,4\},\{1,2,5\},\{1,2,6\},$\\$\{1,2,7\},\{1,2,8\},\{1,2,9\},$\\$\{1,3,4\},\{1,3,5\},\{1,3,6\},\{1,3,7\},$\\$\{1,4,5\},\{1,4,6\},\{1,5,6\},$\\$\{2,3,4\},\{2,3,5\},\{2,3,6\},$\\$\{2,4,5\},\{3,4,5\},\{1,2,3,4,5\}\}$} & \tabincell{c}{$\{\emptyset,\{1\},\{2\},\{3\},$\\$\{4\},\{5\},\{6\},\{7\},$\\$\{8\},\{9\},\{10\},\{11\},$\\$\{12\},\{13\},\{14\},\{15\},$\\$\{1,2\},\{1,3\},\{1,4\},\{1,5\},$\\$\{1,6\},\{1,7\},\{1,8\},$\\$\{2,3\},\{2,4\},\{2,5\},\{2,6\},$\\$\{3,4\},\{3,5\},\{4,5\},$\\$\{1,2,3\},\{1,2,4\},\{1,2,5\},$\\$\{1,3,4\},\{2,3,4\},\{1,2,3,4\}\}$} & 36 \\ \hline
	\end{tabular}
	\label{tab2}
	\vspace{-1em}
\end{table*}

\begin{theorem}
	For $\gamma=0$ and $\mathcal{F} \in \mathcal{F}_{\textnormal{all}}^{\textnormal{even}}$, if $|\mathcal{F}| \geq 2m$, where $m = \lceil(\log_2L+1)/2\rceil$, then $\operatorname{num}(\mathcal{F}) \geq L$ holds. On the other hand, for $\gamma=1$ and $\mathcal{F} \in \mathcal{F}_{\textnormal{all}}^{\textnormal{odd}}$, if $|\mathcal{F}| \geq 2m+1$, where $m = \lceil\log_2L/2\rceil$, then we have  $\operatorname{num}(\mathcal{F}) \geq L$.
	\label{th1}
\end{theorem}
\begin{proof}
	For a flipping combination $\mathcal{F} \in \mathcal{F}_{\textnormal{all}}^{\textnormal{even}}$, it has $\sum_{i=0}^{m}\begin{pmatrix}\begin{smallmatrix} 2i \\ 2m \end{smallmatrix}\end{pmatrix}=2^{2m-1}$ subsets, which means that $\operatorname{num}(\mathcal{F})=2^{2m-1} \geq L$ is satisfied when $m = \lceil(\log_2L+1)/2\rceil$. Likewise, for $\mathcal{F} \in \mathcal{F}_{\textnormal{all}}^{\textnormal{odd}}$, the number of its subsets is $\sum_{i=0}^{m}\begin{pmatrix}\begin{smallmatrix}	2i+1 \\ 2m+1 \end{smallmatrix}\end{pmatrix}=2^{2m}$. Therefore, it holds that $\operatorname{num}(\mathcal{F})=2^{2m} \geq L$ when $m = \lceil\log_2L/2\rceil$. This completes the proof.
\end{proof}

\begin{algorithm}[t]
	\label{alg1}
	\caption{Offline Generation of the MCS for SPC Nodes under $\gamma=0$}
	\LinesNumbered
	\KwIn{$L$}
	\KwOut{$\mathcal{C}$}
	
	$\mathcal{C}=\{\emptyset\}$\;
	\ForEach{$\mathcal{F} \in \mathcal{F}_{\textnormal{all}}^{\textnormal{even}}$}{
		$\operatorname{num}(\mathcal{F})=0$\;
	}
	\For{$i = 1 \to m-1$}{
		$\mathcal{A}=\{\mathcal{F}|\mathcal{F} \in \mathcal{F}_{\textnormal{all}}^{\textnormal{even}},|\mathcal{F}|=2i\}$\;
		\ForEach{$\mathcal{E}_{\mathcal{A}} \in \mathcal{A}$}{
			$\mathcal{B}=\{\mathcal{F}|\mathcal{F} \in \mathcal{F}_{\textnormal{all}}^{\textnormal{even}},|\mathcal{F}| < 2m\}-\mathcal{E}_{\mathcal{A}}$\;
			\ForEach{$\mathcal{E}_{\mathcal{B}} \in \mathcal{T}$}{
				\If{$\mathcal{E}_{\mathcal{A}}$ is less reliable than $\mathcal{E}_{\mathcal{B}}$}{
					\If{$\operatorname{num}(\mathcal{E}_{\mathcal{B}}) \geq L-1$}{
						$\operatorname{num}(\mathcal{E}_{\mathcal{A}})=L$\;
						break\;
					}
					\Else{
						$\operatorname{num}(\mathcal{E}_{\mathcal{A}})=\operatorname{num}(\mathcal{E}_{\mathcal{A}})+1$\;
					}
				}
			}
			\If{$\operatorname{num}(\mathcal{E}_{\mathcal{A}})<L$}{
				$\mathcal{C}=\mathcal{C} \cup \{\mathcal{E}_{\mathcal{A}}\}$\;
			}
		}
	}
	
	\Return $\mathcal{C}$
\end{algorithm}

Then, by combining Theorem~\ref{th1} and Theorems~1, 2, 4 and 5 introduced in \cite{Zhao2021Minimum}, we present in Algorithm~\ref{alg1} the proposed offline MCS generation procedure for SPC nodes under $\gamma=0$. For brevity, we omit the detailed MCS generation algorithm under $\gamma=1$ due to its similarity with Algorithm~\ref{alg1}. Alternatively, a more convenient MCS generation approach for $\gamma=1$ is to directly utilize the results obtained from the case of $\gamma=0$, by taking advantage of the following theorem.

\begin{theorem}
	For any $\mathcal{F}$ in $\mathcal{C}$ associated with $\gamma=0$, define a new flipping combination as $\mathcal{F}' \triangleq \mathcal{F} \triangle \{1\}$, then $\mathcal{F}'$ is included in $\mathcal{C}$ associated with $\gamma=1$.
	\label{th2}
\end{theorem}
\begin{proof}
	For an SPC node with $\gamma = 1$, the ML solution is to flip the least reliable bit, i.e., the bit indicated by the flipping combination $\{1\}$. After flipping this bit, the SPC node is now valid with $\gamma = 0$, and in this case the flipping bits from $\mathcal{F} \in \mathcal{C}$ can keep an even parity check. Therefore, $\mathcal{F}'$ can be derived by flipping the bits indicated by $\{1\}$ and $\mathcal{F} \in \mathcal{C}$. This thus completes the proof.
	\vspace{-0.5em}
\end{proof}

Note that the above construction process is also conversely valid, i.e., the MCS associated with $\gamma=0$ can also be constructed based on the results from the case of $\gamma=1$ using the same transformation as introduced in Theorem~\ref{th2}. Therefore, to derive the MCSes for both cases of $\gamma=0$ and $\gamma=1$, we only need to perform Algorithm~\ref{alg1} once.

For convenience, we list in Table~\ref{tab2} the MCS results under $L=\{2,4,8,16\}$, where $\mathcal{C}^{\textnormal{SPC}}$ and $\mathcal{C}^{\textnormal{R1}}$ denote the MCSes for SPC and R1 nodes, respectively.

\subsection{FPL Decoding} 
\label{sec3.2}
With the help of MCS, the sequential decoding of SPC nodes introduced in \cite{Hashemi2017Fast} can now be parallelized. First, we calculate the value of $\gamma$ (using \eqref{eq11}), based on which we select the corresponding MCS $\mathcal{C}^{\textnormal{SPC}}$ (has already been determined before the decoding process). Then, each parent path will be directly split into several candidate paths by flipping the bits indicated by $\mathcal{F}\in\mathcal{C}^{\textnormal{SPC}}$. Let $\sigma_{\mathcal{F}}$ denote the combination of $|\alpha_p[i]|$ when using $\mathcal{F}$, i.e., the accumulated penalty in PM caused by bit-flipping, which is shown as follows:
\begin{equation}
	\begin{aligned}
		\sigma_{\mathcal{F}} = \sum_{i \in \mathcal{F}}|\alpha_p[(i)_{|\alpha|}]|.
	\end{aligned}
	\label{eq13}
\end{equation}
Accordingly, the PM of the candidate path associated with $\mathcal{F}$ can be obtained by
\begin{equation}
	\begin{aligned}
		\operatorname{PM}_{\mathcal{F}} =& \sum_{i=1}^{N_p}\ln(1+e^{-|\alpha_p[i]|}) + \sigma_{\mathcal{F}}\\
		\approx&\sigma_{\mathcal{F}}, \qquad \textnormal{(HWF)}.
	\end{aligned}
	\label{eq14}
\end{equation}
Finally, through a sorting procedure $|\mathcal{C}| L \to L$, all candidate paths are compared and the $L$ paths with the smallest PMs are reserved. 

\section{Fast List Decoding of SR1/SPC Nodes}
\label{sec4}
In this section, we present a two-stage fast list decoding algorithms for SR1/SPC nodes. In stage I, we propose to decode the source node to correct the P-PCs, while in stage II, the S-PCs are corrected without violating the P-PCs and we present FSL and FPL decoding algorithms for two subtypes of SR1/SPC nodes, i.e., the SR1 and RSR1 nodes.

\subsection{Overview}
\label{sec4.1}
As described in Section~\ref{sec3.4}, the decoded codeword of the SR1/SPC node should satisfy all the P-PCs and S-PCs simultaneously in order to avoid decoding performance degradation. This is also tenable for list decoding, i.e., the codewords of all paths should also satisfy the P-PCs and S-PCs. Therefore, the principle of the proposed fast list decoding algorithm is to produce the $L$ most reliable paths by splitting the previously reserved paths at the node level, while guaranteeing their validity. Similar to our previous work on fast SC decoding \cite{Lu2023Fast}, fast list decoding of the SR1/SPC node can also be divided into two stages, where the P-PCs and S-PCs are corrected successively. For each candidate path, the decoding in stage I temporarily ignores the S-PCs and corrects the P-PCs by decoding the source node, while the decoding in stage II corrects the S-PCs by keeping the already satisfied P-PCs unchanged, such that the final decoding output is valid. Meanwhile, only the $L$ paths with the smallest PMs are retained in each decoding stage. 

\begin{figure*}[ht]
	\centering
	\includegraphics[width=0.8\textwidth]{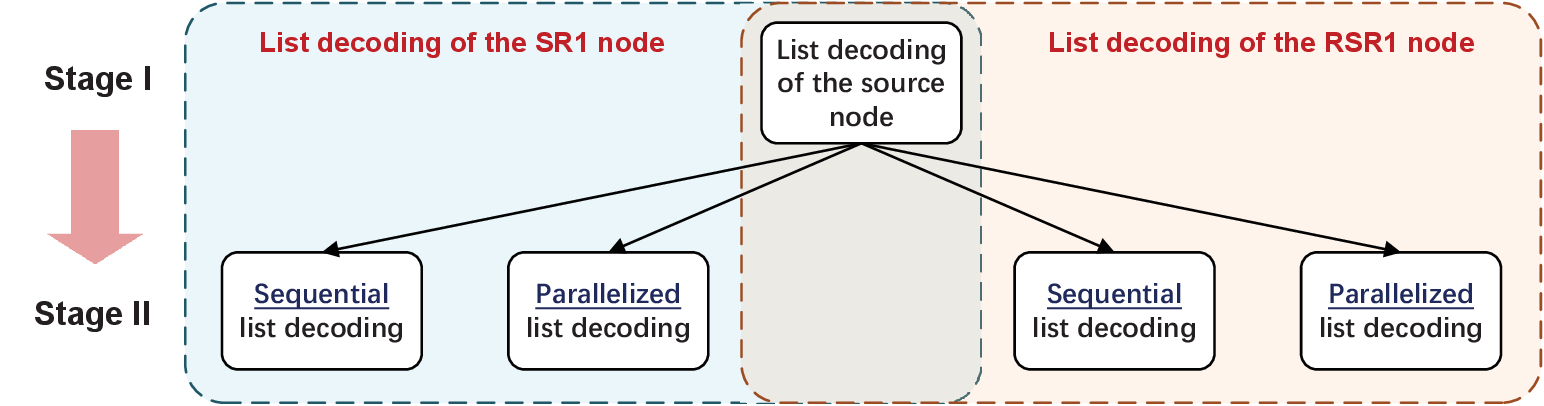}
	\caption{Overview of the proposed list decoding algorithms for SR1/SPC nodes.}
	\label{fig:2}
	\vspace{-1em}
\end{figure*} 

In Fig.~\ref{fig:2}, we depict an overview of the proposed fast list decoding algorithms for SR1/SPC nodes. In particular, considering the trade-off between complexity and latency, we present both FSL and FPL decoding algorithms. These two list decoding algorithms share the same stage I decoding process but vary in stage II. The FSL decoder follows a sequential procedure, which exhibits lower complexity and is more hardware-friendly for practical implementation. In contrast to this, the FPL decoder is highly parallel such that considerable latency reduction can be achieved, but at the cost of increased computational complexity. Besides, since the two subtypes of SR1/SPC nodes, i.e., SR1 and RSR1 nodes, contain different parity checks, the proposed FSL and FPL decoding algorithms are also different. To summarize, one universal stage I and four different stage II decoding algorithms will be introduced in the following.

\subsection{Stage I Decoding}
\label{sec4.2}
First, the hard-decision codeword $\beta_{p}$ needs to be modified to satisfy the P-PCs that are dependent on the codeword of the source node $\beta_{q}$, according to \eqref{eq9}. To do so, we determine the parity check for each SPC subcode, denoted by $\gamma^k_{\textrm{P-PC}}$, as
\begin{equation}
	\begin{aligned}
		\gamma^k_{\textrm{P-PC}} = \bigoplus\limits_{j=1}^{N_p/N_q} \beta_{p}[(j-1)N_q+k] \oplus \beta_{q}[k],
	\end{aligned}
	\label{eq15}
\end{equation}
where the superscript $k$ represents the SPC subcode index introduced in Section~\ref{sec3.4}. Then, each P-PC can be corrected by flipping the least reliable bit with index $(1)^k_{|\alpha|}=\arg\min_{1\leq j \leq N_p/N_q}|\alpha_{p}[(j-1)N_{q}+k]|$ using Wagner decoding \cite{Silverman1954Coding}, where $(j)^k_{|\alpha|}$ is the sorted index which indicates the $j$-th smallest absolute LLR in the $k$-th SPC subcode. Accordingly, the PM of each path is calculated as
\begin{equation}
	\begin{aligned}
		\operatorname{PM}_p =& \sum_{i=1}^{N_p}\ln(1+e^{-|\alpha_p[i]|}) \\
		&+\sum_{k=1}^{N_q}\gamma^k_{\textrm{P-PC}}|\alpha_p[((1)^k_{|\alpha|}-1)N_q+k]|\\
		\approx&\sum_{k=1}^{N_q}\gamma^k_{\textrm{P-PC}}|\alpha_p[((1)^k_{|\alpha|}-1)N_q+k]|, \quad \textnormal{(HWF)}.
	\end{aligned}
	\label{eq16} 
\end{equation}
Finally, all paths are compared, amongst which the $L$ paths with the least PM values are reserved for stage II. 

In the aforementioned decoding process, the critical issue lies in how to determine $\beta_{q}$ in \eqref{eq15} for each SPC subcode. To achieve ML performance, the optimal approach is to estimate each information bit in the source node to obtain $\beta_{q}$, and then compute and compare the PMs of all candidate paths \cite{Ardakani2019Fast}. Suppose that the number of information bits in the source node is $K_q$, then the $L$ parent paths will each generate $2^{K_q}$ candidate paths, which is computationally costly. As an approximate ML solution, SCL decoding can be utilized to obtain a list of potential codewords for the source node. Prior to that, we first calculate the LLRs of the source node as \cite{Lu2023Fast}
\begin{equation}
	\begin{aligned}
		\alpha_{q}[k] = &\prod\limits_{j=1}^{N_p/N_q} \operatorname{sgn}(\alpha_{p}[(j-1)N_{q}+k])\\ &|\alpha_p[((1)^k_{|\alpha|}-1)N_q+k]|.
	\end{aligned}
	\label{eq17}
\end{equation}
Then, the source node is decoded using SCL decoding to produce a list of $\beta_{q}$. However, it is unclear whether a path should be reserved or pruned during the SCL decoding process, due to the unclear relationship between the PM at the source node level, i.e., $\operatorname{PM}_q$, and that at the root node level, i.e., $\operatorname{PM}_p$. For instance, the $L$ smallest PMs at the source node level do not necessarily lead to the $L$ smallest PMs at the root node level. To tackle this problem, we introduce below a theorem which illustrates the relationship between $\operatorname{PM}_q$ and $\operatorname{PM}_p$.

\begin{theorem}
	Suppose that each SPC subcode is decoded using Wagner decoding, then the relationship between $\operatorname{PM}_q$ and $\operatorname{PM}_p$ can be expressed by
	\begin{equation}
		\begin{aligned}
			\operatorname{PM}_p = &\sum_{i\in\mathcal{R}}\ln(1+e^{-|\alpha_p[i]|})+\operatorname{PM}_q\\
			\approx& \operatorname{PM}_q, \qquad \textnormal{(HWF)},
		\end{aligned}
		\label{eq18} 
	\end{equation}
	where $\mathcal{R}$ records all the bit indices except for those least reliable ones in each SPC subcode, i.e., $\mathcal{R}=\{1,2,\ldots,N_p\}-\{((1)^k_{|\alpha|}-1)N_q+k|1\leq k\leq N_q\}$.
	\label{theorem:2}
\end{theorem}
\begin{proof}
	First, based on \eqref{eq17}, we can obtain the absolute value and hard-decision result of $\alpha_{q}[k]$ as follows:
	\begin{equation}
		\begin{gathered}
			|\alpha_{q}[k]| = |\alpha_p[((1)^k_{|\alpha|}-1)N_q+k]|, \\
			\operatorname{HD}(\alpha_{q}[k]) =  \bigoplus\limits_{j=1}^{N_p/N_q} \beta_{p}[(j-1)N_q+k].
		\end{gathered}
		\label{eq19}
	\end{equation}
	By using \eqref{eq19}, the penalty term in \eqref{eq16} can be simplified as
	\begin{equation}
		\begin{aligned}
			&\sum_{k=1}^{N_q}\gamma^k_{\textrm{P-PC}}|\alpha_p[((1)^k_{|\alpha|}-1)N_q+k]|\\
			&= \sum_{k=1}^{N_q}\operatorname{HD}(\alpha_q[k]) \oplus \beta_q[k]|\alpha_q[k]| = \sum_{k\in\mathcal{K}}|\alpha_q[k]|,
		\end{aligned}
		\label{eq20}
	\end{equation}
	where $\mathcal{K} \triangleq \{k|1\leq k\leq N_q, \beta_q[k]\not=\operatorname{HD}(\alpha_q[k])\}$. Then, based on the identity $\ln (1+e^a)-\ln (1+e^{-a})=a$ with $a=|\alpha_q[k]|$, we can rewrite \eqref{eq20} as
	\begin{equation}
		\begin{aligned}
			&\sum_{k=1}^{N_q}\gamma^k_{\textrm{P-PC}}|\alpha_p[((1)^k_{|\alpha|}-1)N_q+k]|\\
			&= \sum_{k\in\mathcal{K}}\ln(1+e^{|\alpha_q[k]|})-\ln(1+e^{-|\alpha_q[k]|})\\
			&\quad+\sum_{k\in\mathcal{K}^c}\ln(1+e^{-|\alpha_q[k]|})-\ln(1+e^{-|\alpha_q[k]|})\\
			&\overset{(a)}{=}\sum_{k=1}^{N_q}\ln(1+e^{-(1-2\beta_q[k])\alpha_q[k]})-\sum_{k=1}^{N_q}\ln(1+e^{-|\alpha_q[k]|})\\
			&=\operatorname{PM}_q-\sum_{k=1}^{N_q}\ln(1+e^{-|\alpha_q[k]|})\\
			&\approx\operatorname{PM}_q, \quad \textnormal{(HWF)},
		\end{aligned}
		\label{eq21}
	\end{equation}
	where $\mathcal{K}^c\triangleq\{1,2,\ldots,N_q\}-\mathcal{K}$ and $(a)$ is derived by resorting to \eqref{eq6}. In particular, applying HWF approximation leads to a simpler result since $\sum_{k=1}^{N_q}\ln(1+e^{-|\alpha_q[k]|})\approx0$. Finally, by replacing the penalty term in \eqref{eq16} by \eqref{eq21}, \eqref{eq18} can be readily proved.
\end{proof}

Theorem~\ref{theorem:2} implies that the $L$ least reliable paths at the root node level can be determined at the source node level. This means that we can directly compare $\operatorname{PM}_p$ and perform path pruning during the SCL decoding process of the source node, and finally obtain a list of favourable paths for the subsequent decoding process. Note that previous fast list decoding techniques can be applied to decode the source node if it exhibits a special structure.

Unless otherwise stated, the subsequent stage II decoding is based on the results obtained from stage I, e.g., the reserved paths and PMs. Accordingly, the parent paths at the beginning of stage II are initialized as the reserved paths in stage I.

\subsection{Stage II Decoding of SR1 Nodes}
\label{sec4.3}
\subsubsection{FSL Decoding} 
\label{sec4.3.1}

First, at the root node level, we calculate the modified LLRs, denoted by $\delta[i]$, as follows:
\begin{equation}
	\begin{aligned}
		\delta[i] = |\alpha_p[i]| + (1-2\gamma^k_{\textrm{P-PC}})|\alpha_p[((1)^k_{|\alpha|}-1)N_q+k]|,
	\end{aligned}
	\label{eq22} 
\end{equation}
where $k=i\%N_q$. Next, we sort $\delta[i]$ in ascending order with a sorted index $(t)_{\delta}$, such that $\delta[(1)_{\delta}] \leq \delta[(2)_{\delta}] \leq \ldots \leq \delta[(N_p-N_q)_{\delta}]$, where $1 \leq t \leq N_p-N_q$. With this sorted index, we serially split the parent paths obtained from stage I. At each step $t$, two temporary candidate paths are generated based on each previously reserved parent path, by considering both hypotheses of flipping and non-flipping events. Specifically, by flipping two bits, i.e., $\beta_p[(t)_{\delta}]$ and $\beta_p[((1)^{k_t}_{|\alpha|}-1)N_q+k_t]$, where $k_t=(t)_{\delta}\%N_q$, a new temporary path can be generated. The other temporary path is exactly the path reserved from the former step as no bit is flipped. Starting from $\operatorname{PM}_{p,0}=\operatorname{PM}_{p}$, the PMs of the temporary candidate paths are updated as
\begin{equation}
	\begin{aligned}
		\operatorname{PM}_{p,t} 
		\!=\! \begin{cases} \operatorname{PM}_{p,t-1}, & \textnormal{if $\beta_p[(t)_{\delta}]\!=\!\operatorname{HD}(\alpha_p[(t)_{\delta}])$} \\ \operatorname{PM}_{p,t-1} + \delta_t, & \textnormal{otherwise}\end{cases},
	\end{aligned}
	\label{eq23} 
\end{equation}
where $\delta_t$ is given by
\begin{equation}
	\begin{aligned}
		\delta_t = |\alpha_p[(t)_{\delta}]| + (1-2\gamma^{k_t}_{\textrm{P-PC}})|\alpha_p[((1)^{k_t}_{|\alpha|}-1)N_q+k]|.
	\end{aligned}
	\label{eq24} 
\end{equation}
Meanwhile, $\gamma^{k_t}_{\textrm{P-PC}}$ is updated as
\begin{equation}
	\begin{aligned}
		\gamma^{k_t}_{\textrm{P-PC}}
		= \begin{cases} \gamma^{k_t}_{\textrm{P-PC}}, & \textnormal{if $\beta_p[(t)_{\delta}]=\operatorname{HD}(\alpha_p[(t)_{\delta}])$} \\ 1-\gamma^{k_t}_{\textrm{P-PC}}, & \textnormal{otherwise}\end{cases}.
	\end{aligned}
	\label{eq25} 
\end{equation}
Then, the $2L$ temporary paths are compared, amongst which the paths with the smallest PMs are reserved for the next step, through a sorting procedure $2L \to L$. Since no S-PC is introduced in SR1 nodes, such a sequential decoding procedure is able to ensure that all parity constraints are satisfied while the candidate paths are the most reliable ones at each step. 

The whole decoding algorithm will terminate in advance when this step is repeated for $\min(L-1, N_p-N_q)$ times, without degrading the decoding performance. Note that when the SR1 node is reduced to a G-PC node, the proposed sequential decoding algorithm will reduce to the fast list decoding introduced in \cite{Ren2022Sequence}.

\subsubsection{FPL Decoding}
\label{sec4.3.2}
Similar to the FPL decoder for SPC nodes in Section~\ref{sec3}, the aforementioned FSL decoder for SR1 nodes can also be highly parallelized based on the MCS as follows. First, an SR1 node can be interpreted as $2^q$ parallel SPC subcodes whose parity checks are obtained via stage I. Then, we decode these parallel SPC subcodes separately, using the proposed FPL decoder introduced in Section~\ref{sec3}. Finally, based on the MCS again, SR1 nodes can be decoded using the decoding results of these SPC subcodes. 

Specifically, for the $k$-th SPC subcode, we first calculate $\gamma^k_{\textrm{P-PC}}$ according to \eqref{eq17}, based on which the corresponding MCS $\mathcal{C}^{\textnormal{SPC}}$ is selected and $\sigma_{\mathcal{F}}$ can be calculated accordingly, where $\mathcal{F}\in\mathcal{C}^{\textnormal{SPC}}$ and $\mathcal{F}\not=\{\emptyset\}$. Specially, $\sigma_{\mathcal{F}}$ (given in \eqref{eq13}) should be properly modified according to $\gamma^k_{\textrm{P-PC}}$, i.e., 
\begin{equation}
	\begin{aligned}
		\sigma_{\mathcal{F}} = &\sum_{j \in \mathcal{F}}|\alpha_p[((j)^k_{|\alpha|}-1)N_q+k]| \\
		&- \gamma^k_{\textrm{P-PC}} |\alpha_p[((1)^k_{|\alpha|}-1)N_q+k]|.
	\end{aligned}
	\label{eq26}
\end{equation} 
Then, using a sorter with radix $N_q |\mathcal{C}|$, all $\sigma_{\mathcal{F}}$ are sorted and only the smallest $L$ ones are reserved to generate the corresponding candidate paths. After the sorting procedure $N_q |\mathcal{C}| \to L$, we assume that the sorted results are arranged in ascending order, i.e., $\sigma_1 \leq \sigma_2 \leq \ldots \leq \sigma_L$. 

However, the resultant candidate paths are not necessarily the most reliable ones, since the SPC nodes are decoded separately without considering their combinations. For example, if $\sigma_1+\sigma_2 < \sigma_L$ holds, then the $L$-th candidate path is no longer the $L$-th most reliable path and thus should be eliminated. In essence, by interpreting all $\sigma_i$ as the absolute LLRs of an R1 node, the MCS $\mathcal{C}^{\textnormal{R1}}$ (see Table~\ref{tab2}) can also be directly employed to select the minimum-combinations of $\sigma_i$, where $1\leq i \leq L$. Let $\Delta_{\mathcal{F}}$ denote the combination of $\sigma_i$ when using $\mathcal{F}\in\mathcal{C}^{\textnormal{R1}}$, i.e.,
\begin{equation}
	\begin{aligned}
		\Delta_{\mathcal{F}} = \sum_{i \in \mathcal{F}}\sigma_i,
	\end{aligned}
	\label{eq27}
\end{equation} 
then the associated PM can be obtained by
\begin{equation}
	\begin{aligned}
		\operatorname{PM}_{\mathcal{F}} = \operatorname{PM}_p + \Delta_{\mathcal{F}}.
	\end{aligned}
	\label{eq28}
\end{equation}
Finally, through a sorting procedure $|\mathcal{C}| L \to L$, all candidate paths are compared and the $L$ paths with the smallest PMs are reserved.

\subsection{Stage II Decoding of RSR1 Nodes} 
\subsubsection{FSL Decoding}
Inspired by the aforementioned sequential path splitting idea, where redundant paths are eliminated from being split during the decoding process, we present in the following an efficient FSL decoding algorithm for RSR1 nodes.

Generally, the proposed FSL decoder follows a similar sequential process as described in Section~\ref{sec4.3.1}, i.e., we split the parent paths step by step, using the ascending modified LLRs $\delta[i]$. As such, the PMs are incremented with $\delta[(t)_{\delta}]$ and corrected P-PCs are kept at each step. However, different from that in Section~\ref{sec4.3.1}, the validity of the temporary candidate paths cannot be guaranteed given the newly introduced S-PCs. Since the whole codeword space will be traversed using the sequential decoding process and the S-PCs can be automatically corrected, we therefore focus on modifying the original pruning strategy to reserve the valid paths instead of correcting the S-PCs. 

A straightforward approach is to eliminate paths according to their PM values. Denote each temporary candidate path as a three-element tuple, including its PM, codeword and validity. Specifically, the PMs and codewords are generated via the sequential decoding process as described in Section~\ref{sec4.3.1}, while the validity, denoted by $\gamma_{\textnormal{S-PC}}$, is calculated according to \eqref{eq10}:
\begin{equation}
	\begin{aligned}
		\gamma_{\textnormal{S-PC}} = \bigoplus\limits_{j=1}^{N_q/N_r/2} \bigoplus\limits_{k=1}^{N_r} \beta_{p}[(2j-1)N_r+k]. 
	\end{aligned}
	\label{eq29}
\end{equation}
At each step, we record the $L$-th smallest PM of all valid paths after path splitting and sorting. Then, we can safely eliminate the paths whose PMs are larger than this PM value, regardless of their validity. Note that the invalid paths with smaller PMs should also be reserved, since these paths may turn to be valid in the subsequent steps. At the final step, we terminate the whole decoding process by selecting the $L$ valid paths with the smallest PMs as the output.

To further reduce the computational complexity, we propose a modified pruning strategy, where only the paths that must not lead to redundant paths are split. At step $t$, a threshold, denoted by $\Delta_t$, is first calculated as
\begin{equation}
	\begin{aligned}
		\Delta_t = \widetilde{\operatorname{PM}}_{L,t-1}-\widetilde{\operatorname{PM}}_{1,t-1},
	\end{aligned}
	\label{eq30}
\end{equation}
where $\widetilde{\operatorname{PM}}_{l,t-1}$ represents the $l$-th smallest PM of the valid temporary candidate paths at step $t-1$. Then, the threshold $\Delta_t$ is applied on each path to check if this path should be split into two threads, or remain reserved, or be eliminated, which is shown as follows:
\begin{equation}
	\begin{aligned}
		&\textnormal{If $\delta[(t)_{\delta}] \leq \Delta_t$, then split the path as in \eqref{eq23}},\\
		&\textnormal{if $\delta[(t)_{\delta}] > \Delta_t$ and $\gamma_{\textnormal{S-PC}}=0$, then $\operatorname{PM}_{p,t} = \operatorname{PM}_{p,t-1}$}, \\
		&\textnormal{if $\delta[(t)_{\delta}] > \Delta_t$ and $\gamma_{\textnormal{S-PC}}=1$, then $\operatorname{PM}_{p,t} = +\infty$}, \\
	\end{aligned}
	\label{eq31} 
\end{equation}
where $\operatorname{PM}_{p,t} = +\infty$ indicates that this path should be directly eliminated from further decoding process. The whole path splitting procedure will be terminated when the number of the remaining paths is equal to $L$. Let $\tau$ denote the value of $t$ when the decoding algorithm is terminated, then $\tau$ is a variable whose value depends on the code parameters and channel condition.

\subsubsection{FPL Decoding}
Different from the SR1 nodes, there are additional S-PCs imposed on RSR1 nodes and it is difficult to directly apply the MCS-based method for FPL decoding of RSR1 nodes. To address this issue, we employ the FCS $\mathcal{S}$ introduced in our previous work \cite{Lu2023Fast}, which contains a number of flip coordinates each indicates two flipping bits, i.e., $\beta_{p}[(j_1-1)N_q+k]$ and $\beta_{p}[(j_2-1)N_q+k]$, and is denoted by $\mathcal{E}=(j_1,j_2,k)$. Considering that all the P-PCs have been corrected in stage I, we can utilize the flipping coordinates in $\mathcal{S}$ to correct the remaining S-PCs and thereby keep all the parity constraints satisfied. Specifically, to devise a list decoding algorithm, a direct but also heuristic approach is to generate a list of valid candidate paths by splitting each parent path using the FCS. Therefore, among the $|\mathcal{S}| L$ candidate paths, the $L$ paths with the smallest PMs are reserved, where $|\mathcal{S}|$ is the size of the FCS. However, these candidate paths are not necessarily the most reliable ones. Besides, the flip coordinates in $\mathcal{S}$ are not able to provide candidate paths for an existing path whose S-PCs are satisfied automatically. Therefore, employing this heuristic approach may lead to decoding performance degradation.

To tackle this problem, we propose to pre-process the parent paths, instead of directly splitting them using the FCS. First, using the sorted modified LLRs $\delta[i]$ (calculated by \eqref{eq22}), we split each parent path to $\upsilon$ additional candidate paths such that
\begin{equation}
	\begin{aligned}
		\operatorname{PM}^{o+1}_p
		= \operatorname{PM}_p + \delta[(o)_{\delta}],
	\end{aligned}
	\label{eq32}
\end{equation}
where $\operatorname{PM}^{1}_p=\operatorname{PM}_p$ presents the parent path, $\operatorname{PM}^{o}_p$ denotes the PMs of the obtained candidate paths and $1\leq o\leq \upsilon$ is the candidate path index. Since the S-PCs of each candidate path is not necessarily satisfied, we then further split each existing candidate path and obtain $(\upsilon+1) |\mathcal{S}|$ new paths with all parity constraints satisfied. The detailed path splitting procedure is similar to the fast SC decoding algorithm of SR1/SPC node with the aid of FCS (cf. Section~IV.~C in \cite{Lu2023Fast}), which is omitted here for brevity. Accordingly, the PMs of these new candidate paths are obtained by
\begin{equation}
	\begin{aligned}
		\operatorname{PM}^o_{\varepsilon}
		= \operatorname{PM}^o_p + \lambda_{\mathcal{E}},
	\end{aligned}
	\label{eq33}
\end{equation}
where $\lambda_{\mathcal{E}}$ is given by \cite{Lu2023Fast}
\begin{equation}
	\begin{aligned}
		\lambda_{\mathcal{E}} = &\sum_{j=j_1,j_2}(1-2 \beta_{p}[(j-1)N_q+k]) \alpha_{p}[(j-1)N_q+k].
	\end{aligned}
	\label{eq34}
\end{equation}
Finally, the reserved paths are determined by selecting the $L$ paths with the smallest PMs from $(\upsilon+1) |\mathcal{S}| L$ candidate paths. By using a sorter with large radix (e.g., radix-64 sorter for $\upsilon=3$, $|\mathcal{S}|=4$ and $L=4$), the proposed FPL decoder can achieve extremely low decoding latency.

\section{Simulation Results}
\label{sec5}
In this section, the decoding latency and error-correction performance of the proposed fast list decoders are compared with the state-of-the-art ones. Throughout this section, we consider CRC-aided polar codes \cite{Niu2012CRC}, denoted by $\mathcal{P}(N,K,r)$, where $r$ is the number of CRC bits. Unless otherwise specified, the considered polar codes are constructed using the Gaussian approximation method introduced in \cite{Trifonov2012Efficient}. Generally, the state-of-the-art fast list decoders presented in \cite{Hashemi2017Fast,Ardakani2019Fast,Zhao2021Minimum,Ren2022Sequence} are abbreviated as the \emph{SOTA decoders} in the following. As this paper mainly focuses on high SPC and SR1/SPC nodes, the decoding of the other special nodes (not covered by the SR1/SPC nodes) follows \cite{Hashemi2017Fast,Zhao2021Minimum,Ren2022Sequence}, where the fast list decoders for R0 and REP nodes are from \cite{Hashemi2017Fast}, the FSL and FPL decoders for R1 nodes are from \cite{Zhao2021Minimum}, and the fast list decoder for SR0/REP nodes is from \cite{Hashemi2017Fast} and \cite{Ren2022Sequence}. Combining the conventional CA-SCL decoder \cite{Niu2012CRC} with these special node decoders and together with the decoders proposed in this paper, the resultant decoder is referred to as the \emph{proposed decoder (with FSL/FPL)} hereafter, where the notation \emph{(with FSL/FPL)} indicates that the high-rate special nodes, including R1, SPC and SR1/SPC nodes, are all decoded by the FSL or FPL decoders.

\subsection{Decoding Latency Analysis}
In this subsection, we measure the decoding latency of various decoders by counting the required number of time steps, under the following assumptions \cite{Ardakani2019Fast,Zheng2021Threshold}. First, there is no limitation on hardware resources such that all the parallelizable operations can be carried out in one time step. Second, addition/subtraction of real numbers (e.g., the check-node operations) consume one time step. Third, the hard decision and bit operations can be performed instantly, without consuming any additional time steps. Last, we consider the following two approaches for sorting:
\begin{itemize}
	\item \emph{Full-rank sorters} \cite{Zhao2021Minimum,Ren2022Sequence}: The considered decoders support large-radix sorters such that the smallest $L$ PMs can be selected in one time step. 
	\item \emph{Pipeline-layered sorters} \cite{Ardakani2019Fast}: We only consider the standard $2L \to L$ sorter as the basic sorting unit. This means a large-radix sorter should be decomposed into multiple pipeline-layered $2L \to L$ sorters such that $\log_2|\mathcal{C}|$ time steps are required to sort $|\mathcal{C}| L$ PMs.
\end{itemize}
For the FPL decoders, the full-rank sorters enable the highest level of parallelism and lead to the best case in terms of decoding latency, whereas the pipeline-layered sorters, on the flip side, lead to the worst case. For clarify, we use the notations \emph{FPL-F} and \emph{FPL-P} to differentiate the FPL decoders when employing the full-rank and pipeline-layered sorters, respectively. 

\begin{remark}
	In practice, the full-rank sorters can be optimized for lower complexity. One of the most popular way is to prune the sorter if some input data have already been sorted in advance. Typically, the PMs of the $L$ parent paths are usually arranged in ascending order before path splitting, and thus the comparison between these PMs can be released from the $2L \to L$ sorting procedure \cite{Balatsoukas2015LLR}. In this work, apart from the rank orders provided by the PMs of the parent paths, we can also acquire some additional rank orders with the aid of MCS. Specifically, some candidate paths are intrinsically more reliable than the others according to their flip combinations, which means that the order of their PMs can be pre-determined. For instance, given $\mathcal{F}_1=\{1,2\}$ and $\mathcal{F}_2=\{2,3\}$ for an SPC node with $\gamma=0$ (see Table~\ref{tab2}), we naturally have $\operatorname{PM}_{\mathcal{F}_1}<\operatorname{PM}_{\mathcal{F}_2}$ according to \eqref{eq13} and \eqref{eq14}. This property can be employed to reduce the number of comparisons and further lower the computational complexity significantly. In case there is no prior rank order results, another approach is to design a partial-rank sorter that partially sorts the PMs, as done in \cite{Ren2022Sequence}, which is able to roughly halve the computational complexity. 
\end{remark}

\begin{table*}[t]
	\caption{Required Number of Time Steps to Decode Different Special Nodes}
	\centering
	\begin{threeparttable}
		\begin{tabular}{c|c|c|c|c}\hline
			& \multicolumn{4}{c}{SOTA decoders} \\ \cline{2-5}
			& TSP'17 \cite{Hashemi2017Fast} & TCOM'19 \cite{Ardakani2019Fast} & CL'21 \cite{Zhao2021Minimum} & TSP'22 \cite{Ren2022Sequence} \\ \hline

			R0 & 1 & 1 & 1 & 1\\ \hline
			
			REP & 2 & 2 & 2 & 2\\ \hline
			
			R1 & $\min(L-1,N_{p})$ & $\min(L-1,N_{p})$ & 1 & $\min(L-1,N_{p})$\\ \hline
			
			\textbf{SPC} & $\min(L,N_p)$ & $\min(L,N_p)$ & $\min(L,N_p)$ & $\min(L,N_p)$\\ \hline
					
			\textbf{SR1/SPC} 
			& \tabincell{c}{$T_q+2(p-q)$\\$+\sum_{s\in\mathcal{L}}\min(L,N_s)$\\$+\sum_{s\in\mathcal{L}^c}\min(L-1,N_s)$} & \tabincell{c}{$T_q+2(p-q)$\\$+\sum_{s\in\mathcal{L}}\min(L,N_s)$\\$+\sum_{s\in\mathcal{L}^c}\min(L-1,N_s)$\tnote{$\dagger$}} & \tabincell{c}{$T_q+2(p-q)+|\mathcal{L}^c|$\\$+\sum_{s\in\mathcal{L}}\min(L,N_s)$} & \tabincell{c}{$T_q+2(p-q)$\\$+\sum_{s\in\mathcal{L}}\min(L,N_s)$\\$+\sum_{s\in\mathcal{L}^c}\min(L-1,N_s)$\tnote{$\ddagger$}} \\ \hline
		\end{tabular}
		\begin{tablenotes}  
			\footnotesize
			\item[$\dagger$] When the SR1/SPC node is a Type III or Type IV node, this number will be reduced to $1+\min(L-1,N_{p}-2)$ or $1+\min(L-1,N_{p}-4)$, respectively.   
			\item[$\ddagger$] When the SR1/SPC node is a G-PC node, this number will be reduced to $1+\min(L-1,N_{p}-N_{q})$. 
		\end{tablenotes} 
		\vspace{1em}
	\end{threeparttable}
	\centering
	\begin{threeparttable}
		\begin{tabular}{c|c|c|c|c}\hline
			\multicolumn{2}{c|}{} & \multicolumn{3}{c}{Proposed decoders} \\ \cline{3-5}
			\multicolumn{2}{c|}{} & w/ FSL & w/ FPL-F & w/ FPL-P \\ \hline

			\multicolumn{2}{c|}{R1} & $1+\min(L-1,N_p)$\cite{Hashemi2017Fast} & $1$\cite{Zhao2021Minimum} & $\log_2|\mathcal{C}|$ \\ \hline			
			
			\multicolumn{2}{c|}{\textbf{SPC}} & $\min(L,N_p)$ & $1$ & $\log_2|\mathcal{C}|$ \\ \hline
			
			\multirow{2}{*}{\textbf{SR1/SPC}} 
			&
			\tabincell{c}{\textbf{SR1}} & \tabincell{c}{$T_q+1+\min(L-1,N_p-N_q)$} & \tabincell{c}{$T_q+3$} & \tabincell{c}{$T_q+1+\log_2N_q|\mathcal{C}|/L+\log_2|\mathcal{C}|$} \\ \cline{2-5}
			
			& \tabincell{c}{\textbf{RSR1}} & \tabincell{c}{$T_q+1+\tau$} & \tabincell{c}{$T_q+4$} & \tabincell{c}{$T_q+3+\log_2(1+\upsilon)|\mathcal{S}|$} \\ \hline
		\end{tabular}
	\end{threeparttable}
	\label{tab3}
	\vspace{-0.5em}
\end{table*}

Based on these assumptions, we provide the following latency analysis for the considered special node decoders. Note that the required time steps to decode SPC and SR1/SPC nodes depends on the employed list decoders.
\begin{itemize}
	\item SOTA decoder: The SOTA SPC decoder follows a sequential decoding procedure and thus consumes $1+\min(L-1,N_{p}-N_{q})$ time steps \cite{Hashemi2017Fast}. For an SR1/SPC node which cannot be directly decoded by the existing SOTA decoders, the source node and the descendant R1 or SPC nodes should be serially decoded following the decoding tree depicted in Fig.~\ref{fig:1}. Therefore, the total number of time steps can be calculated by adding the decoding latency of these nodes and the additional latency required by the check-node operations, which can be expressed by\footnote{In this case, the SOTA fast decoders introduced in \cite{Hashemi2017Fast,Ardakani2019Fast,Ren2022Sequence} are used.}
	\begin{equation*}
		\begin{aligned}
			T_q+2(p-q)+\sum_{s\in\mathcal{L}}\min(L,N_s)+\sum_{s\in\mathcal{L}^c}\min(L-1,N_s),
		\end{aligned}
	\end{equation*}
	where $T_q$ is the number of time steps to decode the source node. In particular, since the SOTA decoder in \cite{Zhao2021Minimum} is able to decode an R1 node in one time step by using an FPL-F decoder, this number can be reduced to
	\begin{equation*}
		\begin{aligned}
			T_q+2(p-q)+|\mathcal{L}^c|+\sum_{s\in\mathcal{L}}\min(L,N_s).
		\end{aligned}
	\end{equation*}
	\item Proposed decoder (with FSL): As mentioned above, the FSL decoder requires $1+\min(L-1,N_p-1)$ time steps to decode an SPC node \cite{Hashemi2017Fast}. For SR1 nodes, the FSL decoder consumes $T_q+1$ time steps in stage I, where the ``+1'' time step is for the LLR calculation in \eqref{eq17}. Besides, each decoding step in stage II requires one time step for path splitting and PM update, resulting in additional $\min(L-1,N_p-N_q)$ time steps. To sum up, the total number of time steps required to decode an SR1 node is $T_q+1+\min(L-1,N_p-N_q)$. Likewise, the proposed FSL decoder for RSR1 nodes consumes $T_q+1+\tau$ time steps in total, where $\tau$ time steps are required for the sequential decoding procedure in stage II. 
	\item Proposed decoder (with FPL-F): First, the proposed FPL decoder for SPC nodes can generate all candidate paths in one round of path splitting, such that only one time step is required to select $L$ candidate paths from $|\mathcal{C}| L$ potential ones. For an SR1 node, since it is decoded as a group of parallel SPC subcodes and then as an R1 node in stage II, two time steps are required in total, where one time step is for decoding the parallel SPC subcodes and the other one is for decoding the R1 node. For an RSR1 node, three time steps are required with two time steps to calculate the PMs in \eqref{eq32} and \eqref{eq33} and one time step to select $L$ candidate paths from $(\upsilon+1)|\mathcal{S}| L$ ones.
	\item Proposed decoder (with FPL-P): When employing the pipelined-layered sorting approach, more time steps are required by each large-radix sorter and the decoding latency is thus related with the sorter's radix, i.e., the number of the PMs to be sorted after path splitting. Following the above analysis of the proposed decoder (with FPL-F), the decoding latency can be re-calculated accordingly. 
\end{itemize}
In general, the latency reduction of the proposed FSL decoders stems from the savings of check-node operations and the separate decoding of the descendant nodes. On the other hand, the proposed FPL decoders can achieve significant decoding speedup by simplifying and parallelizing the sequential path splitting procedure, with the help of large-radix sorters. 

\subsection{Comparison with Existing Works}
In this subsection, we compare the decoding latency and error-correction performance of the SOTA decoders \cite{Hashemi2017Fast,Ardakani2019Fast,Zhao2021Minimum,Ren2022Sequence} with the proposed decoders. 
\subsubsection{Decoding Latency}
\begin{figure*}[!ht]
	\centering
	\setlength{\abovecaptionskip}{-0.3em}
	\includegraphics[width=0.9\textwidth]{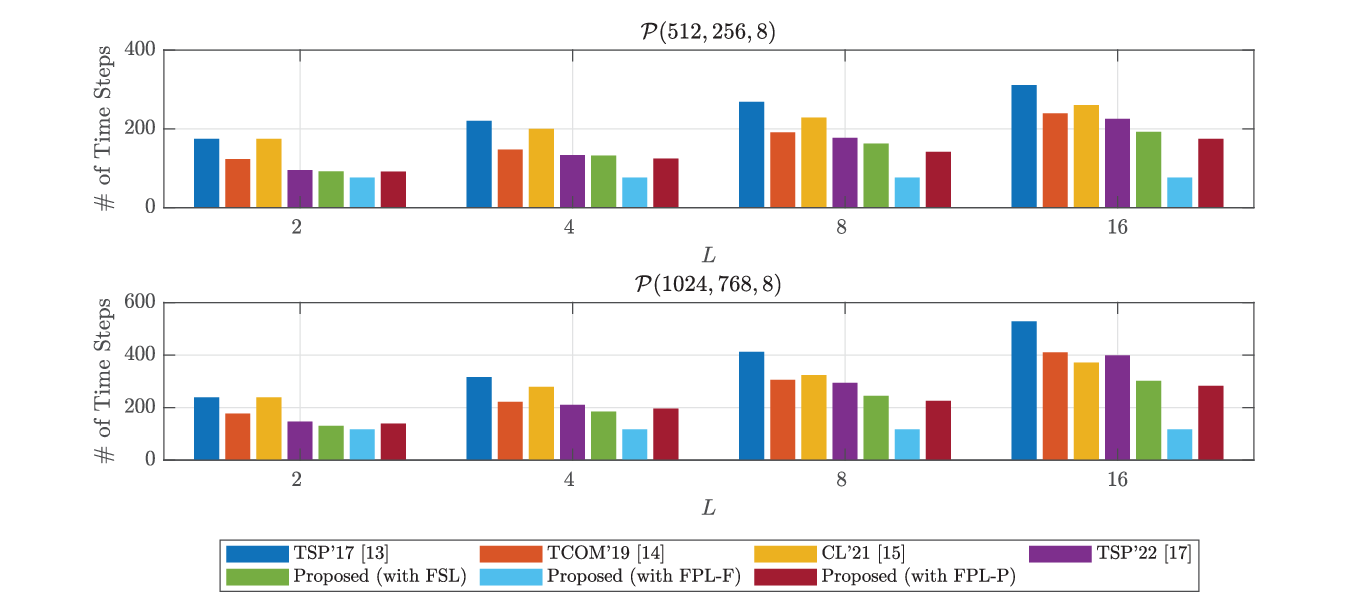}
	\caption{Required time step number comparison for the SOTA decoders \cite{Hashemi2017Fast,Ardakani2019Fast,Zhao2021Minimum,Ren2022Sequence} and the proposed decoder (with FSL/FPL).}
	\label{fig:3}
	\vspace{-1em}
\end{figure*} 

First, Fig.~\ref{fig:3} compares the required time step numbers for the considered decoders, where the code parameters are $\mathcal{P}(512,256,8)$ and $\mathcal{P}(1024,768,8)$ with $L=\{2,4,8,16\}$. Note that $\tau$ is counted at $E_b/N_0=2.0$ dB and $\upsilon$ is set to $\upsilon=L-1$. It can be observed that the proposed decoders require fewer time steps as compared to the SOTA decoders especially when the code rate is high. Amongst the considered decoders, the proposed decoder (with FPL-F) achieves the lowest decoding latency thanks to the highly parallelized path splitting procedure aided by large-radix sorters. Besides, except for the proposed decoder (with FPL-F), all the other decoders consume more time steps as the list size $L$ increases, which is consistent with the analysis in Table~\ref{tab3}. Furthermore, one can see that generally the time saving achieved by the proposed decoders tends to increase as the code length or list size becomes larger. By adopting the proposed decoder (with FPL-F), the decoding latency of the SOTA decoder in \cite{Ren2022Sequence} can be reduced by 66.2\%
and 70.7\% for $\mathcal{P}(512,256,8)$ and $\mathcal{P}(1024,768,8)$ with $L = 16$, respectively. In the other extreme, when employing the FPL-P decoders for lower complexity, up to 29.3\% decoding latency can be saved for $\mathcal{P}(1024,768,8)$ with $L = 16$.

\subsubsection{Error-Correction Performance}
\begin{figure}[!ht]
	\centering
	\setlength{\abovecaptionskip}{-0.2em}
	\includegraphics[width=0.47\textwidth]{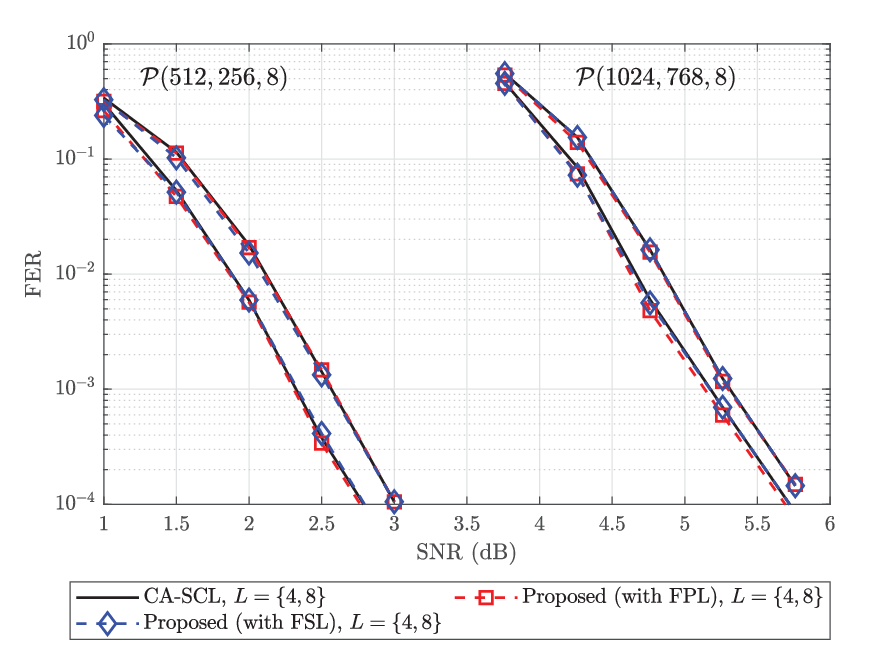}
	\caption{FER performance of the conventional CA-SCL decoder and the proposed decoder (with FSL/FPL), where $L=\{4,8\}$.}
	\label{fig:4}
	\vspace{-1em}
\end{figure} 
Then, we depict in Fig.~\ref{fig:4} the frame error rate (FER) performance of the conventional CA-SCL decoder and the proposed decoders, for $\mathcal{P}(512,256,8)$ and $\mathcal{P}(1024,768,8)$ with $L=\{4,8\}$. It can be observed that the FER performance is preserved when using the proposed decoders. 

\subsection{Empirical Optimizations}
\label{sec5.3}
In this subsection, we introduce two empirical optimization methods for the proposed FSL and FPL decoders, which can be used to further reduce the decoding complexity (in terms of either time, space or computation).

\subsubsection{Reducing the Number of Path Splitting}
\label{sec5.3.1}
All FSL decoders for high-rate special nodes need to perform path splitting for the bits in the root node, until all these bits are traversed (i.e., path splitting are performed for $N_p$ times). However, as mentioned in Section~\ref{sec4.3.1}, the number of path splitting can be limited to a certain value to preserve the error-correction performance, while reducing the latency and computational complexity caused by the redundant path splitting steps. In practice, we can further reduce this number at the cost of minor performance degradation. Based on some empirical simulations, different values can be employed as an upper limit on the number of path splitting (denoted as $T_{\textnormal{max}}$), which provides the flexibility to trade some error-correction performance for higher decoding speed and lower decoding complexity.

\begin{figure}[t]
	\centering
	\includegraphics[width=0.48\textwidth]{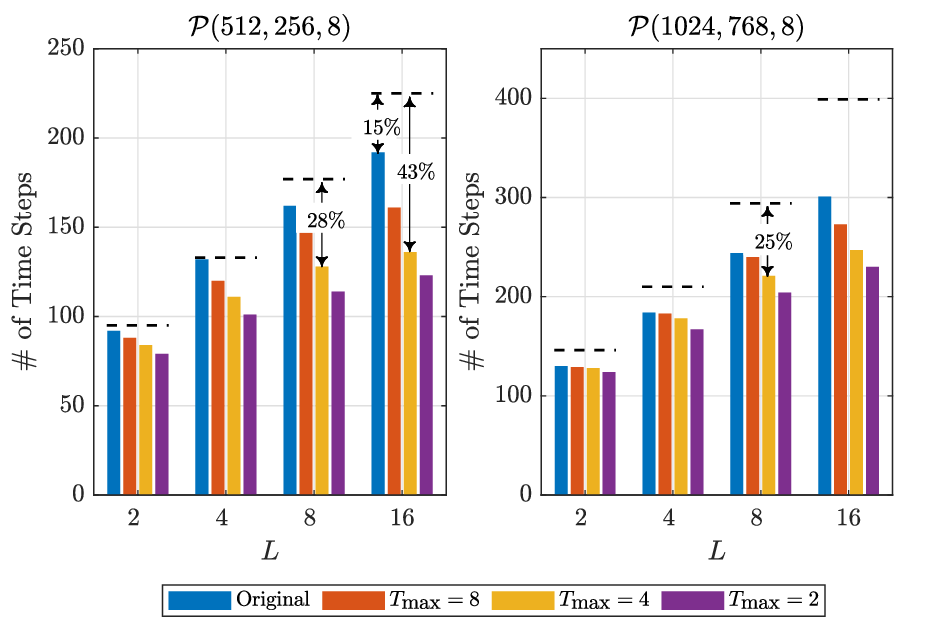}
	\caption{Decoding latency of the proposed decoder (with FSL) for different values of $T_{\textnormal{max}}$, where the black dashed lines represent the decoding latency of the SOTA decoder proposed in \cite{Ren2022Sequence}.}
	\label{fig:5}
	\vspace{-1em}
\end{figure} 

Fig.~\ref{fig:5} exhibits how $T_{\textnormal{max}}$ impacts the decoding latency of the proposed decoder (with FSL), where ``Original'' is the case without empirical optimization. It can be observed that reducing the number of path splitting results in different degrees of latency reduction. In particular, the decoding speedup with respect to the SOTA decoder is increased from 15\% to 43\% when setting $T_{\textnormal{max}}=4$, for $\mathcal{P}(512,256,8)$ with $L=16$. However, when the code length is larger, the speedup advantage diminishes in terms of percentage since the total decoding latency is increased. 

\begin{figure}[t]
	\centering
	\setlength{\abovecaptionskip}{-0.2em}
	\includegraphics[width=0.45\textwidth]{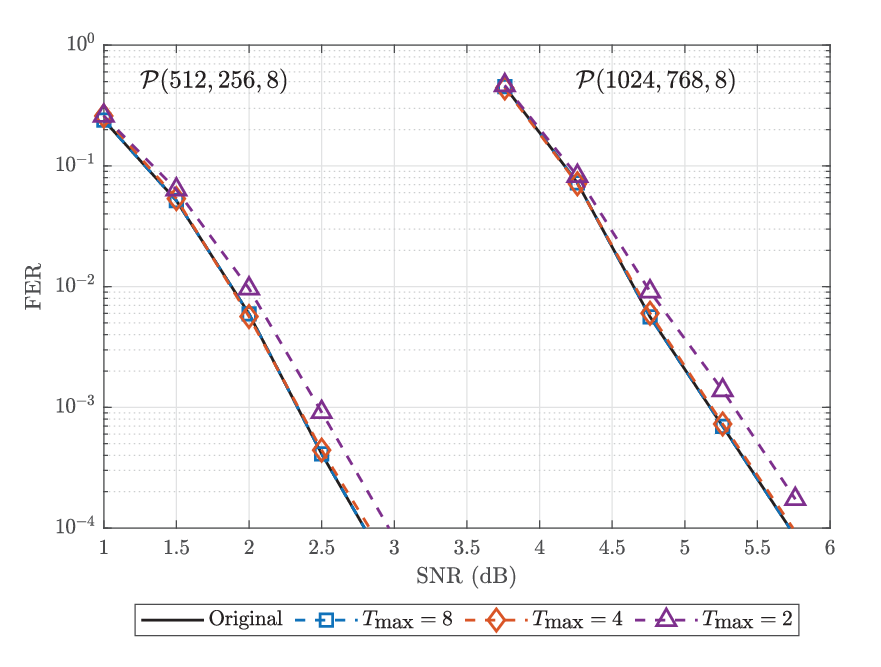}
	\caption{FER performance of the proposed decoder (with FSL) for different values of $T_{\textnormal{max}}$, where $L=8$.}
	\label{fig:6}
	\vspace{-1em}
\end{figure} 

To further investigate the impact of this optimization method on the decoding performance, we depict in Fig.~\ref{fig:6} the FER performance of the proposed decoder (with FSL) under different values of $T_{\textnormal{max}}$, for $\mathcal{P}(512,256,8)$ and $\mathcal{P}(1024,768,8)$ with $L=8$. As shown in Fig.~\ref{fig:6}, selecting $T_{\textnormal{max}}=\{4,8\}$ results in negligible performance degradation, while achieving considerable speedup advantages according to Fig.~\ref{fig:5}. When selecting $T_{\textnormal{max}}=2$, the decoding latency can be further reduced at the cost of about 0.2 dB performance loss. This means $T_{\textnormal{max}}=4$ is an appropriate value for the proposed FSL decoders with $L=8$, based on which the decoding latency as compared to the SOTA decoder proposed in \cite{Ren2022Sequence} can be reduced by 28\% and 25\% for $\mathcal{P}(512,256,8)$ and $\mathcal{P}(1024,768,8)$, respectively, while the error-correction performance is preserved.

\subsubsection{Restricting the Size of Pre-determined Sets}
\label{sec5.3.2}
All FPL decoders need to pre-determine some index sets, e.g., the MCS and FCS, to record the unreliable bit positions that should be eliminated from path splitting. In practice, employing these sets consumes extra memory space, and may also lead to higher computational complexity as more data are required to be compared through sorters. To achieve lower computational and space complexity, the sizes of these sets can be restricted to an empirical value. Typically, we halve the MCS size by properly excluding some unreliable flipping combinations. First, we set a threshold $I_{\textnormal{max}}$ given a specific list size $L$. Then, a flipping combination $\mathcal{E}$ is considered to be unreliable and thus should be eliminated, if it meets the condition $\sum_{i\in\mathcal{F}}>I_{\textnormal{max}}$. For instance, by setting $I_{\textnormal{max}}=6$ for $\mathcal{C}^{\textnormal{SPC}}$ with $\gamma=0$ (see Table~\ref{tab2}), we can obtain the optimized MCS as $\mathcal{C}^{\textnormal{SPC}}=\{\emptyset,\{1,2\},\{1,3\},\{1,4\},\{1,5\},\{2,3\},\{2,4\}\}$ with $|\mathcal{C}^{\textnormal{SPC}}|=7$. By carefully selecting the value of $I_{\textnormal{max}}$, the sizes of the MCSes can all be halved successfully. 

\begin{figure}[t]
	\centering
	\includegraphics[width=0.48\textwidth]{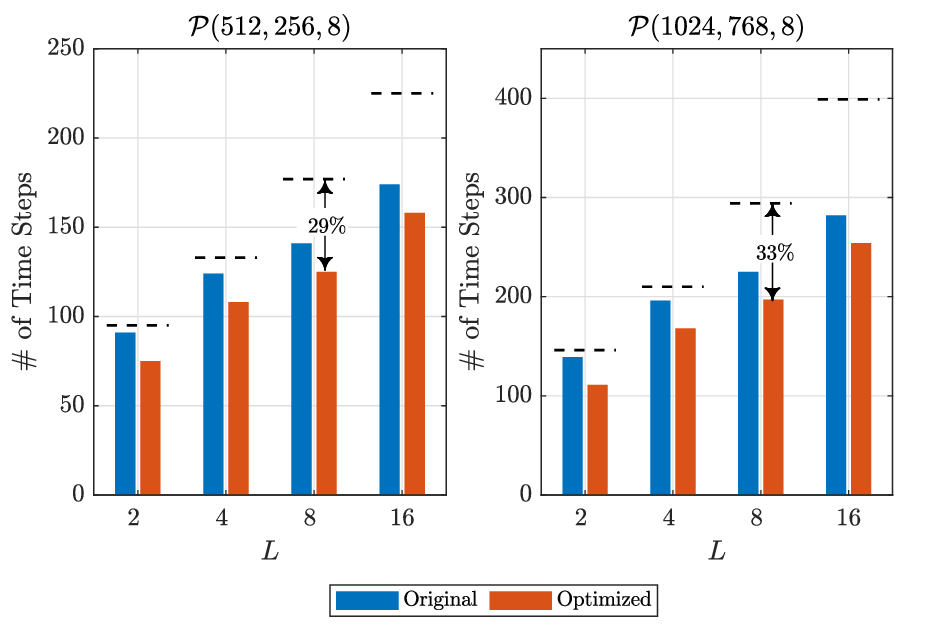}
	\caption{Decoding latency of the proposed decoder (with FPL-P) with or without the proposed optimization method, where the black dashed lines represent the decoding latency of the SOTA decoder proposed in \cite{Ren2022Sequence}.}
	\label{fig:7}
	\vspace{-1.5em}
\end{figure} 

\begin{figure}[!ht]
	\centering
	\setlength{\abovecaptionskip}{-0.5em}
	\includegraphics[width=0.45\textwidth]{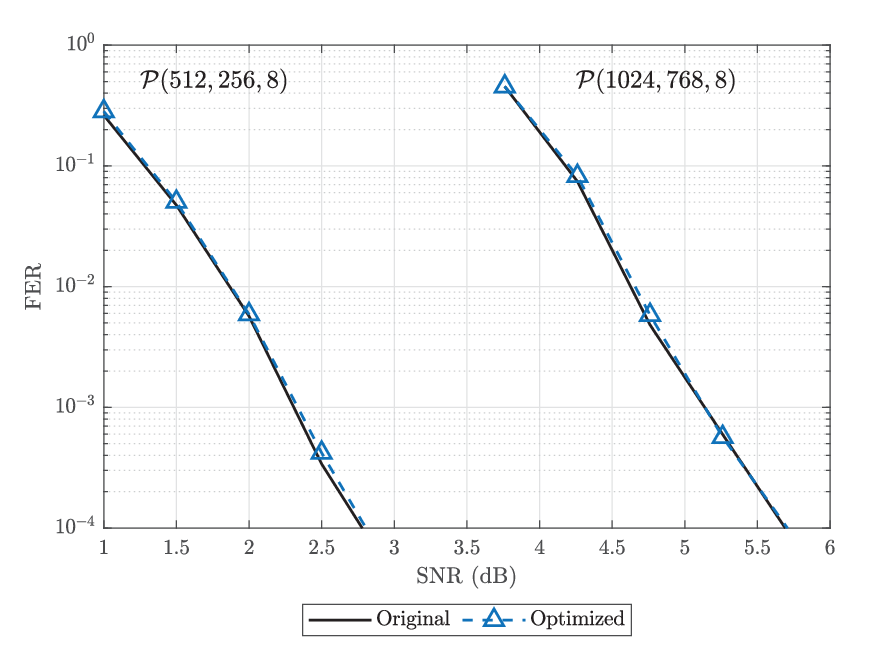}
	\caption{FER performance of the proposed decoder (with FPL) with or without the proposed optimization method, where $L=8$.}
	\label{fig:8}
	\vspace{-0.5em}
\end{figure} 

Similarly, we depict in Fig.~\ref{fig:7} and Fig.~\ref{fig:8} the decoding latency and FER performance of the proposed decoder (with FPL-P) when using the proposed optimization method, respectively. Combining the results in Fig.~\ref{fig:7} and Fig.~\ref{fig:8}, it can be observed that halving MCS size can further reduce the decoding latency without inducing any performance loss. In particular, the latency reduction with respect to the SOTA decoder is roughly 30\% for the considered polar codes, which amounts to the speedup advantage achieved by the proposed decoder (with FSL) (see Fig.~\ref{fig:5}). Note that the speedup achieved by the proposed decoder (with FPL-P) is measured under the worst case where only the pipeline-layered sorting approach is considered. Since the MCS size can be halved, it is more convenient to support large-radix sorters for the FPL decoders. As such, the decoding latency can be further reduced to approach the lower bound achieved by the proposed decoder (with FPL-F) (shown in Fig.~\ref{fig:3}). To summarize, although the whole search space of high-rate special nodes can be narrowed down to the proposed MCSes theoretically, it can be further shrunk by employing the  proposed empirical optimization methods to make the proposed FPL decoders more efficient.

\section{Conclusions}
\label{sec6}
In this work, we proposed fast list decoders to significantly reduce the SCL decoding latency of high-rate polar codes (SPC and SR1/SPC nodes). For SPC nodes, we showed how to parallelize the conventional sequential path splitting procedure by introducing the MCS to pre-determine the redundant paths in advance. For SR1/SPC nodes, we presented two decoders, i.e., the FSL and FPL decoders, to provide achieve a flexible tradeoff between decoding latency and complexity. Compared with the SOTA fast list decoders, the proposed decoders can preserve the error-correction performance yet with considerable lower decoding latency.


{
\bibliographystyle{IEEEtran}
\bibliography{IEEEabrv,mybibfile}
}

\end{document}